\documentclass[acmsmall,nonacm]{acmart}

\setcopyright{rightsretained}
\acmPrice{}
\acmDOI{10.1145/3434336}
\acmYear{2021}
\copyrightyear{2021}
\acmSubmissionID{popl21main-p556-p}
\acmJournal{PACMPL}
\acmVolume{5}
\acmNumber{POPL}
\acmArticle{55}
\acmMonth{1}

\usepackage{booktabs}   
\usepackage{subcaption} 
\usepackage{savesym}    
\usepackage[utf8]{inputenc}
\usepackage[T1]{fontenc}
\usepackage{microtype}
\usepackage{tabularx}
\usepackage{enumerate}
\usepackage{stmaryrd}
\usepackage{paralist}
\usepackage{hyperref}
\usepackage{prftree}
\usepackage{relsize}
\usepackage{apxproof} 
\usepackage{tikz-cd}
\usepackage{esvect}
\usepackage{listings}
\hypersetup{
    citecolor=brown,
    filecolor=black,
    linkcolor=black,
    urlcolor=brown
}
\savesymbol{comment}
\usepackage[nompar,nosign,draft]{commenting}
\declareauthor{sr}{Steven}{blue}
\authorcommand{sr}{comment}
\declareauthor{ej}{Eddie}{red}
\authorcommand{ej}{comment}

\theoremstyle{acmdefinition}
\newtheorem{remark}{Remark}

\newtheoremrep{theorem}{Theorem}
\newtheoremrep{lemma}{Lemma}

\newif\ifsupp

\newcommand{\Map}{\mathsf{Map}}

\newcommand{\Ty}{\mathsf{Ty}}
\newcommand{\Dt}{\mathsf{Dt}}
\newcommand{\ETy}{\mathsf{ETy}}

\newcommand{\ul}[1]{\underline{#1}}
\newcommand{\restr}[2]{#1\!\!\restriction_{#2}}

\newcommand{\types}{\vdash}

\renewcommand{\implies}{\supset}

\newcommand{\mng}[1]{\llbracket #1 \rrbracket}

\newcommand{\frv}{\mathsf{FRV}}
\newcommand{\subtype}{\sqsubseteq}

\newcommand{\abra}[1]{\langle #1 \rangle}

\newcommand{\arity}{\mathsf{Arity}}
\newcommand{\dom}{\mathsf{dom}}

\newcommand{\matchtm}[2]{\mathsf{case}\ #1\ \mathsf{of}\ \{#2\}}

\newcommand{\abs}[2]{\lambda #1.\,#2}
\newcommand{\Abs}[2]{\Lambda #1.\,#2}

\newcommand{\fresh}{\mathsf{Fresh}}

\newcommand{\Sat}{\mathsf{Sat}}

\newcommand{\Con}{\mathbb{K}}

\newcommand{\infers}{\Longrightarrow}

\newcommand{\rlnm}[1]{%
  \textsf{(#1)}
}

\newcommand{\subenv}[2]{\abra{#1}_{#2}}

\newcommand{\inj}{\mathsf{inj}}

\newcommand{\Sch}{\mathsf{Sch}}

\newcommand{\canon}[1]{#1^{\!*}}

\newcommand{\under}{\mathcal{U}}

\newcommand{\bananas}[1]{\llparenthesis #1 \rrparenthesis}

\makeatletter
\def\term#1{%
    \@term#1 \@empty
}
\def\@term#1 #2{%
   \mathit{#1}
   \ifx #2\@empty\else
    \:\expandafter\@term  
   \fi
   #2%
}
\makeatother

\usepackage{pict2e}

\makeatletter
\DeclareRobustCommand{\mmid}{%
  \mathrel{\mathpalette\med@mid\relax}%
}
\newcommand{\med@mid}[2]{%
  \sbox\z@{$\m@th#1\vdash$}%
  \setlength{\unitlength}{\ht\z@}%
  \begin{picture}(.3,1)
  \roundcap
  \linethickness{%
    \ifdim\unitlength<0.9ex
      0.08%
    \else
    \ifdim\unitlength<1.2ex
      0.07%
    \else
      0.06%
    \fi
  \fi\unitlength}
  \polyline(.2,0.03)(.2,0.97)
  \end{picture}%
}

\newtheoremstyle{TheoremNum}
        {\topsep}{\topsep}              
        {\itshape}                      
        {}                              
        {\bfseries}                     
        {.}                             
        { }                             
        {\thmname{#1}\thmnote{ \bfseries #3}}
    \theoremstyle{TheoremNum}

\begin{document}

\title{Intensional Datatype Refinement}         
\subtitle{With Application to Scalable Verification of Pattern-Match Safety}                     


\author{Eddie Jones}
\affiliation{
  \department{Department of Computer Science}              
  \institution{University of Bristol}            
 \country{UK}                    
}

\author{Steven Ramsay}
\affiliation{
  \department{Department of Computer Science}             
  \institution{University of Bristol}           
  \country{UK}                   
}

\begin{abstract}
  The pattern-match safety problem is to verify that a given functional program will never crash due to non-exhaustive patterns in its function definitions.
  We present a refinement type system that can be used to solve this problem.
  The system extends ML-style type systems with algebraic datatypes by a limited form of structural subtyping and environment-level intersection.
  We describe a fully automatic, sound and complete type inference procedure for this system which, under reasonable assumptions, is worst-case linear-time in the program size.
  Compositionality is essential to obtaining this complexity guarantee.
  A prototype implementation for Haskell is able to analyse a selection of packages from the Hackage database in a few hundred milliseconds.
\end{abstract}



\maketitle


\section{Introduction}\label{sec:intro}

\begin{figure}
    \begin{tabular}{p{.3\columnwidth}p{.65\columnwidth}}
      \begin{lstlisting}
  data L a = 
      Atom a 
      | NegAtom a
        
  data Fm a = 
      Lit (L a)
      | Not (Fm a)
      | And (Fm a) (Fm a)
      | Or (Fm a) (Fm a)
      | Imp (Fm a) (Fm a)
      \end{lstlisting}&
      \begin{lstlisting}
  nnf (Lit (Atom x)) = Lit (Atom x)
  nnf (Lit (NegAtom x)) = Lit (NegAtom x)
  nnf (And p q) = And (nnf p) (nnf q)
  nnf (Or p q) = Or (nnf p) (nnf q)
  nnf (Imp p q) = Or (nnf (Not p)) (nnf q)
  nnf (Not (Not p)) = nnf p 
  nnf (Not (And p q)) = Or (nnf (Not p)) (nnf (Not q))
  nnf (Not (Or p q)) = And (nnf (Not p)) (nnf (Not q))
  nnf (Not (Imp p q)) = And (nnf p) (nnf (Not q))
  nnf (Not (Lit (Atom x))) = Lit (NegAtom x)
  nnf (Not (Lit (NegAtom x))) = Lit (Atom x)
  
  nnf2dnf (Lit a) = [[a]]
  nnf2dnf (Or p q) = List.union (nnf2dnf p) (nnf2dnf q)
  nnf2dnf (And p q) = distrib (nnf2dnf p) (nnf2dnf q)
    where distrib xss yss = 
            List.nub [ List.union xs ys | xs <- xss, ys <- yss ]
      \end{lstlisting}
    \end{tabular}\vspace{-15pt}
  \caption{Conversion to disjunctive normal form.}\label{fig:dnf-ex}
  \end{figure}

The \emph{pattern match safety} problem asks, given a program with non-exhaustive (algebraic datatype) patterns in its function definitions, is it possible that the program crashes with a pattern-match exception?
Consider the example Haskell code in Figure \ref{fig:dnf-ex}.
This code defines the two main ingredients in a typical definition (see e.g. \cite{harrison-handbook}) of conversion from arbitrary propositional formulas to propositional formulas in disjunctive normal form (represented as lists of lists of literals).
Using these definitions, the conversion can be described as the composition $\mathsf{dnf} \coloneqq \mathsf{nnf2dnf} \circ \mathsf{nnf}$.

Notice that the definition of $\mathsf{nnf2dnf}$ is partial: it is expected only to be used on inputs that are in negation normal form (NNF).
Consequently, unless is can be verified that $\mathsf{nnf}$ always produces a formula without any occurrence of $\mathsf{Imp}$ or $\mathsf{Not}$, then any application of $\mathsf{dnf}$ to an expression of type $\mathsf{Fm}\ a$ may result in a pattern match failure exception.
In this paper we present a new refinement type system that can be used to perform this verification statically and automatically.
Type inference is compositional and incremental so that it can be integrated with modern development environments: open program expressions can be analysed and only the parts of the code that are modified need to be re-analysed as changes are made.

\paragraph{Contributions.}
Whilst there are other analyses in the literature that can also verify instances of the foregoing example ours is, as far as we are aware, the only to offer strong guarantees on predictability, which we believe to be key to the usability of such systems in practice.
\begin{itemize}
    \item The analysis is characterised by the type system, which is a natural, yet expressive extension of ML-style type systems with algebraic datatypes.  It combines polyvariance (through environment-level intersection) and path-sensitivity (through conditional match typing).  
    \item The analysis runs in time that is, in the worst-case, linear in the size of the program (under reasonable assumptions on the size of types and the nesting of matching).
\end{itemize}
We do not know of any other system or reachability analysis combining \emph{polyvariance}, \emph{path-sensitivity}, an intuitive characterisation of completeness, and a linear-time guarantee on the overall worst-case complexity (in terms of program size).
Furthermore, our prototype demonstrates excellent performance over a range of packages from Hackage, processing each in less than a second.

\subsection{A Type System for Intensional Datatype Refinements}

Sound and terminating program analyses are conservative: there are always programs without bugs that, nevertheless, cannot be verified.  
Identifying a large fragment for which the analysis is complete, i.e. a class of safe programs for which verification is guaranteed, allows the programmer to reason about the behaviour of the analysis on their code. 
In particular, when an analysis fails to verify a program that the user believes to be safe, it gives them an opportunity to take action, such as by programming more defensively, in order to put their program into the fragment and thus be certain of verification success.

However, for this to be most effective, the fragment must be easily understood by the average functional programmer.
Our analysis is complete with respect to programs typable in a natural extension of ML-style type systems with algebraic datatypes. 
Indeed it is characterised by this system: the force of Theorems \ref{thm:inf-mod-correctness} and \ref{thm:equi-consistency} is to say that it forms a sound and complete inference procedure. 
The system is presented in full in Section \ref{sec:type-assignment}, but the highlights are as follows:
\begin{enumerate}[(i)]
  \item The datatype environment introduced by the programmer, e.g. $\mathsf{L}\,a$ and $\mathsf{Fm}\,a$, is \emph{completed}: every datatype whose definition can be obtained by erasing constructors from one of those given is added to the environment for the purpose of type assignment.  These new datatypes are called \emph{intensional refinements}.  These additional types allow for the scrutinee of a match to be typed with a datatype that is more precise  than the underlying type provided by the programmer.  For example, the datatypes in Figure \ref{fig:fm-refns} are among the intensional refinements of $\mathsf{Fm}\ a$, where \lstinline{data A a = Atom a} is an intensional refinement of $\mathsf{L}\, a$.
  Of course, the names of the datatypes are irrelevant.
  
  \item There is a natural notion of subtyping between intensional refinement datatypes which is incorporated into the type system through an unrestricted subsumption rule.  For example, $\mathsf{Clause}\, a$ and $\mathsf{Cube}\,a$ are both subtypes of the intensional refinement:
  \begin{lstlisting}
      data NFm = Lit (L a) | Or (NFm a) (NFm a) | And (NFm a) (NFm a)
  \end{lstlisting}
  which is itself a subtype of $\mathsf{Fm}\, a$.  However, $\mathsf{Clause}\, a$, $\mathsf{Cube}\,a$ and $\mathsf{STLC}\,a$ are all incomparable.

  \item The typing rule for the case analysis construct, by which pattern matching is represented, enforces that matching is exhaustive with respect to the type of the scrutinee.  This ensures that the analysis of matching is sound: programs for which the match is not exhaustive will not be typable.  Moreover, the rule is \emph{path-sensitive}, with the type of the match only depending on the types of the branches corresponding to the type of the scrutinee.  For example, the following function can be assigned the type $(a \to b) \to \mathsf{Cube}\,a \to \mathsf{Cube}\,b$ and it can be assigned the type $(a \to b) \to \mathsf{Clause}\,a \to \mathsf{Clause}\,b$, but not the type $(a \to b) \to \mathsf{STLC}\,a \to \mathsf{STLC}\,b$ because it does not handle the constructor $\mathsf{Imp}$.
  \begin{lstlisting}
      map f (Lit (Atom x))   = Lit (Atom (f x))
      map f (Lit (NegAtom x)) = Lit (NegAtom (f x))
      map f (And p q) = And (map p) (map q)
      map f (Or p q)  = Or (map p) (map q)
  \end{lstlisting}
  Path sensitivity is essential for handling typical use cases.  
  Often a single large datatype is defined but, locally, certain parts of the program work within a fragment (e.g. only on clauses).  
  Path sensitivity helps to ensure that transformations on values inside the fragment remain inside the correct datatype refinement --- otherwise $\mathsf{map}$ could only advertise that it returns formulas in type $\mathsf{NFm}\,a$.
  For example, Elm-style web applications typically define a single, global datatype of actions although the constituent pages may only be prepared to handle certain (overlapping) subsets locally.
  
  \item Finally, refinement polymorphism, and hence \emph{context-sensitivity}, is provided by allowing for environments that have more than a single refinement type binding for each free program variable, i.e. an environment-level intersection.  For example, suppose $\mathsf{trivial} : \mathsf{Clause}\,a \to \mathsf{Bool}$ checks a clause for complementary literals, $\mathsf{sing} : \mathsf{Cube}\, a \to \mathsf{Bool}$ checks if a cube consists of a single conjunct, and $\mathsf{rn} : \mathsf{String} \to \mathsf{String}$ performs a renaming of propositional atoms.  Then the following expression\footnote{The example is rather contrived, but we may rather imagine such combinations occurring in different parts of the program.} is well typed:
  \[
    \abs{xy}{\mathsf{trivial}\,(\mathsf{map}\,\mathsf{rn}\,x) \mathrel{||} \mathsf{sing}\,(\mathsf{map}\,\mathsf{rn}\,y) }
  \]
  This is because the typing environment contains \emph{both} of the aforementioned types for $\mathsf{map}$.
  Note: this is polymorphism in the class of formulas, not only in the type $a$ of their atoms.
\end{enumerate}
To distinguish between the typing assigned to the program by the programming language (which we consider part of the input to the analysis) from the types that can be assigned in our extended system, we call the former the \emph{underlying} typing of the program.

\begin{figure}
  \begin{tabular}{p{.3\columnwidth}p{.3\columnwidth}p{.3\columnwidth}}
    \begin{lstlisting}
data Clause a = 
    Lit (L a) 
    | Or (Clause a) (Clause a)
    \end{lstlisting}&
    \begin{lstlisting}
data STLC a =
    Lit (A a)
    | And (STLC a) (STLC a)
    | Imp (STLC a) (STLC a)
    \end{lstlisting}&
    \begin{lstlisting}
data Cube a =
    Lit (L a)
    | And (Cube a) (Cube a)
    \end{lstlisting}
  \end{tabular}\vspace{-15pt}
  \caption{Some intensional refinements of $\mathsf{Fm}\ a$.}\label{fig:fm-refns}
  \end{figure}

Characterising the analysis with a type system allows the programmer to reason about its behaviour using typings as a kind of \emph{certificate}.  
Returning to the above example, the programmer can be certain that uses of $\mathsf{dnf}$ will be verifiably safe because they can synthesize the intensional datatype refinement $\mathsf{NFm}$, and check the typings $\mathsf{nnf} : \mathsf{Fm}\,a \to \mathsf{NFm}\,a$ and $\mathsf{nnf2dnf} : \mathsf{NFm}\,a \to [\![\mathsf{L}\,a]\!]$.

\subsection{Compositionality and Complexity}

Our analysis takes the form of a type inference procedure for the system described above.
As is typical, inference proceeds by generating and solving typing constraints.
The constraints are guarded inclusions, representing flow of data conditioned on the presence of certain constructors in datatypes along a program path.

A key goal of our work is to give some \emph{guarantee} of the scalability of the analysis to large, real-world programs. 
We do this by ensuring that the whole of type inference -- constraint generation and constraint solving -- runs in time that is worst-case linear in the size of the program (assuming other parameters, such as the size of underlying types, are fixed).

We achieve this complexity guarantee by a careful exploitation of compositionality in the type inference algorithm.
The key is to ensure that the size of the constraint set used to summarise the behaviour of a component $c$ is independent of the number of components that $c$ depends on.

The issues involved are the same for any kind of compositional program analysis so, to illustrate, consider some abstract program $p = p_1 \cdot p_2 \cdot \cdots \cdot p_n$ that has been broken down into $n$ ``components'' $p_i$.  In the interests of approaching the worst-case complexity in as simple a way as possible, assume that each component uses only the component immediately preceding it in the chain --- for example, via a procedure call.

A compositional program analysis computes a summary $S_i$ of the behaviour of each component $i$ separately.   
For example, for constraint-based analyses, this is typically a set of constraints in a solved form (e.g. a constrained type scheme).
For each component, the size of $S_i$ can depend on the size of component $i$ (i.e. its program text), but also the size of the summary $S_{i-1}$ already computed for component $i-1$ on which it depends.
By choosing the granularity of components to be small, or otherwise by making some reasonable assumption, we can regard the size of the program text of each component to be bounded by a constant.  
Hence, when we speak of program size, we will refer to the number of components, $n$.
A consequence of this is that we may assume that the number of times that $p_i$ uses $p_{i-1}$ is bounded by some constant, say $c$.

When analysing the worst-case complexity of polyvariant analyses, like HM(X)-style type inference \cite{odersky-sulzmann-wehr-TSPOS1999},  there is typically the possibility that the summary of component $i$ may be duplicated $c$ times inside the summaries of those components that depend on it,  and thus we arrive at the (well-known) conclusion that the ``summary'' for the entry point of the chain $S_n$ may be of size exponential in $n$\footnote{However, note that \citet{RN52} show that this can be reduced to cubic complexity in the case of simple variable/variable constraints.}.
For non-polyvariant, non-path-sensitive analyses, there is no duplication, but it is nevertheless typical that summaries are already quadratic in $n$: the cubic-time fragment of set constraints (see e.g. \cite{RN66,RN49,RN63,RN64}) is one example of this class.

Since this blow-up occurs even in typical cases, there is an extensive literature on powerful simplification techniques by which large summaries may \emph{sometimes} be replaced by more concise equivalents, see particularly \cite{RN54,RN69,RN70,RN50,RN68,RN38,RN49}.  
However, getting just the right combination and tuning of heuristics is difficult, and the initial implementation effort and subsequent maintenance is significant (e.g. regular benchmarking as the underlying programming language evolves).
Moreover, one will always be able to find reasonable programs on which heuristic simplifications are not well tuned, the program analysis/type inference will stall, and the program's author will lose faith in the system.

By contrast, our system is designed to guarantee that the worst-case size of any $S_i$ is independent of the summaries that it depends on, and hence of the program size (though it is exponential in the size of the largest underlying type).
Therefore, with other parameters fixed, the size of each our summaries $S_i$ is bounded by a constant.

We proceed component by component, first generating constraints and immediately putting them into a solved form.
However, computing a solved form so as to guarantee the above property is not straightforward. 
Our constraint solver, which is inspired by the resolution-based approach used in set constraint based program analysis \cite{aiken-et-al-popl1994,aiken-wimmers-fpca1993,aiken-wimmers-lics1992,heintze-et-al-jar1992} may take time exponential in the size of its input.

We are able to guarantee a linear time complexity overall because our compositional approach ensures that each constraint set that is given to the solver is unrelated to the size of the program.
The size of the constraint set generated for a given component depends only on the size of the \emph{summaries} of the components it depends on --- the solved forms --- and each of these is bounded by a constant.
Therefore, the size of any constraint set supplied to the solver is also bounded by a constant.
Thus we solve a small (but exponential in the size of the underlying types) number of constraints at every program point, rather than an enormous (exponential in the size of the program) number of constraints when processing the program's entry point.

This works only because we show that our constraint sets in solved form have the following remarkable property, stated formally as Theorem~\ref{thm:restr-ext}.
\begin{quotation}
Suppose $C$ is a set of constraints in solved form over variables $V$ and let $I \subseteq V$ be  arbitrary.  
Let $\restr{C}{I}$, called the \emph{restriction of $C$ to $I$} be those constraints in $C$ in which occur \emph{only} variables from $I$. 
Then every solution to $\restr{C}{I}$ can be extended to a solution of all $C$.
\end{quotation}
In the restriction $\restr{C}{I}$, entire constraints are culled, including those that involve a mixture of variables from $I$ and $V \setminus I$. 
Such mixed constraints, intuitively, impose compatibility requirements on the different components of a solution to $C$.
What is significant about the above property is that it guarantees not only that the part of the constraint set only concerned with $V \setminus I$ is internally consistent but, moreover, that the mixed constraints will be satisfiable no matter which solution to $\restr{C}{I}$ is chosen.

We exploit compositionality in order to choose a minimal set of variables $I$, \emph{the interface}, with which to restrict constraint sets.  
The interface of a program expression in context $\Gamma \types e : T$ consists only of the those refinement variables that occur free in $\Gamma$ and $T$.
The size of the interface depends only on the size of the underlying type of $e$, the size of definitions of any datatypes occurring in that type and the nesting of pattern matching.  Thus, if we make the (in our view, reasonable) assumption that the sizes of these quantities are bounded by a constant, so too is the size of the interface and, therefore, the size of any restricted constraint set --- our component summary.

\subsection{Implementation}

Of course worst-case complexity is only part of the story, and especially so when the constant factors depend upon several assumptions.  Hence, we have implemented our System in Haskell as a GHC Plugin and ran it on a selection of packages from the Hackage database.  
The plugin takes a Haskell package to be compiled and runs our type inference algorithm over the whole code to yield a constrained type assignment and a set of type errors.
The average time taken to process each module is in the order of milliseconds and the results show very stark contrast between the number of refinement variables associated with the program points in the module (often > 10000) and the number of refinement variables in the interfaces (typically < 20).

\subsection{Outline}

The rest of the paper is structured as follows.
In Section~\ref{sec:types} we describe a Haskell-like functional programming language which forms the setting for our work.  
This is followed in Section~\ref{sec:refinement} by our definitions of refinement.
Then in Sections~\ref{sec:subtyping} and \ref{sec:type-assignment} by the definition of the type system that characterises the analysis.
In Sections~\ref{sec:constraints}, \ref{sec:inference} and \ref{sec:saturation} we present our analysis as a type inference algorithm, generating and solving constraints.
We discuss the restriction operation and its complexity in Section~\ref{sec:analysis} and we report on our implementation in Section~\ref{sec:implementation}.
Finally, we conclude and discuss related work in Section~\ref{sec:conclusion}.

\section{Language}\label{sec:types}

\paragraph{Preliminaries.}
Given sets $X$ and $Y$, let us write $X \to Y$ for the set of all functions from $X$ to $Y$ and $\Map\ X\ Y$ for the set of all finite maps between $X$ and $Y$.
As usual function arrows are assumed to associate to the right.
Additionally, we define the indexing of function arguments, that is $(X_1 \to \cdots \to X_m \to Y)[i] = X_i$ for all $i \in [1 .. m]$.
Given a family of sets $Y_x$ indexed by $x \in X$, let us write $\Pi x \in X.\ Y_x$ for the subset of $X \to \bigcup_{x \in X} Y_x$ that contains only functions that are guaranteed to map each $x \in X$ to some element of $Y_x$ and let us write $\Sigma x \in X.\ Y_x$ for the subset of $X \times \bigcup_{x \in X} Y_x$ in which the second component of each pair $(x,y)$ is guaranteed to belong to $Y_x$.
Given a family of sets $Y_x$ indexed by $x \in X$, let us write $\coprod_{x \in X} Y_x$ for their disjoint sum and $
\inj_x$ for each of the canonical injections.

\paragraph{Types.}
We assume a countable collection $\mathbb{A}$ of type variables, ranged over by $\alpha$; a finite collection $\mathbb{B}$ of base types, ranged over by $b$, and a countable collection $\mathbb{D}$ of \emph{algebraic datatype identifiers} ranged over by $d$.
These can be thought as the names of first-order type constructors.
Each datatype identifier has a fixed arity, and only forms a proper type when supplied with the appropriate number of type arguments.
We refer to a datatype identifier with its argument as a datatype, and when it is clear from the context we will also write these as $d$.
\[
  \begin{array}{rrcl}
    \textsc{(Monotypes)} & T,U,V &\Coloneqq& \alpha \mid b \mid d\ \vv{T} \mid T_1 \to T_2 \\
    \textsc{(Type schemes)} & S &\Coloneqq& T \mid \forall \alpha.\, S\\
  \end{array}
\]
We write $\Dt\ D$ to stand for the set of all datatypes with datatype identifiers drawn from the set $D$.
$\Ty\ D$ and $\Sch\ D$ are defined similarly for monotypes and schemes.
We consider monotypes to be a trivial instance of type schemes where convenient.
The purpose of distinguishing base types from datatypes is that the former may not be refined.
For example, we will consider $\mathsf{Int}$ to be a base type, $\mathsf{Fm}$ a datatype identifier and $\mathsf{Fm\ Int}$ a datatype.
Type schemes are identified up to renaming of bound variables.

\paragraph{Lifting over types.}
Given a relation on datatypes $R \subseteq \Dt\ D_1 \times \Dt\ D_2$, we write $\Ty(R)$ for the relation on $\Ty\ D_1 \times \Ty\ D_2$ defined inductively by the following:
\[
  \begin{array}{c}
  \prftree{\phantom{hi}}{\Ty(R)(\alpha,\,\alpha)} \qquad
  \prftree{\Ty(R)(b,\,b)} \qquad
    \prftree[r]{$R(d_1,\,d_2)$}{\Ty(R)(d_1,\,d_2)} \qquad
  \prftree{\Ty(R)(T_3,\,T_1)}{\Ty(R)(T_2,\,T_4)}{\Ty(R)(T_1 \to T_2,\,T_3 \to T_4)} 
  \end{array}
\]
\begin{toappendix}
We write $\Sch(R)(\forall \vv{\alpha}.\ T_1,\,\forall \vv{\alpha}.\ T_2)$ just if $\Ty(R)(T_1,\,T_2)$.
Given a function $f: \Dt\ D_1 \to \Dt\ D_2$ we define $\Ty(f) : Ty\ D_1 \to \Ty\ D_2$ recursively as follows:
\[
  \begin{array}{rcl}
    \Ty(f)(\alpha) &=& \alpha \\
    \Ty(f)(b) &=& b \\
    \Ty(f)(d) &=& f(d) \\
    \Ty(f)(T_1 \to T_2) &=& \Ty(f)(T_1) \to \Ty(f)(T_2)
  \end{array}
\]
For brevity, for function $f : \Dt\ D_1 \to \Dt\ D_2$, we will usually just write $f(T)$, $f(S)$ for $\Ty(f)(T)$ and $\Ty(f)(S)$ respectively.
\end{toappendix}

\paragraph{Expressions and modules.}
We assume a countable collection of term variables, ranged over by $x$, $y$, $z$, and variations.
As well as a countable collection $\mathbb{K}$ of datatype constructors, ranged over by $k$.
The arity of a constructor is denoted $\arity(k)$.
The expressions of the language are:
\[
  \begin{array}{rcl}
    m &\Coloneqq& \epsilon \mid m \cdot \langle x = e \rangle \\
    e &\Coloneqq& c \mid k \mid e_1\ e_2 \mid e\ T \mid \abs{x:T}{e} \mid \Abs{\alpha}{e} \\ 
      &  & \mid \matchtm{e}{k_1\ \vec{x_1} \mapsto e_1 \shortmid \cdots{} \shortmid k_n\ \vec{x_n} \mapsto e_n} 
  \end{array}
\]
Expressions are identified up to renaming of bound variables and we will adopt the Barendregt variable convention in order to retain a simple notation.
Since we are defining a refinement type system, we will assume that the input program already has a typing assigned by the underlying type system of the programming language.
We assume that this is manifest, in part, by the insertion of appropriate type abstraction $\Abs{\alpha}{e}$ and application $e\ T$ terms, and by the annotation of term abstractions with their argument type $\abs{x:T}{e}$.
We also assume, as is the case for GHC Core, that pattern matching has been preprocessed into case expressions in which patterns are 1 level deep, i.e have the form $k\,x_1\,\cdots\,x_n$ for some constructor $k$.
Since our analysis is path sensitive, this is not a restriction and we could perform an equivalent, but clumsier, development allowing for syntax containing nested patterns.

Modules $m$ are simply a sequence of variable definitions $\langle x = e \rangle$ that may be empty $\epsilon$.
For simplicity of presentation, we allow recursive definitions but not mutually recursive definition sets.

\paragraph{Datatype environments.}
The meaning of datatypes is defined by an environment of datatype definitions.
Each datatype definition introduces a new datatype identifier along with a collection of datatype constructors that can be used to build instances of the type.

\begin{definition}[Datatype Environment]
  A \emph{datatype environment} is a pair consisting of a set $D \subseteq \mathbb{D}$ of datatypes identifiers and a function $\Delta : D \to \Map\ \Con\ (\Sch\ D)$ mapping each datatype identifier $d$ to a finite map which records the associated constructors and their type.
  We assume these types only concern the type variables that appear in datatype's definition, and so is of shape: $\forall \vv{\alpha}.\ U_1 \to \cdots \to U_m \to d\ \vv{\alpha}$, where $m$ is the constructor's arity.
  For convenience, and as the return type of a constructor is predetermined, we will often identify $\Delta(d)(k)$ with just the sequence of types corresponding to a constructor's arguments, i.e. $[U_1,\, ..,\, U_m]$ where $U_i = \Delta(d)(k)[i]$.

Since datatype environments are partial functions on $\mathbb{D}$, they inherit the natural partial order in which $\Delta_1 \subseteq \Delta_2$ just if the graph of the former is included in the graph of the latter.
If $\Delta_1 \subseteq \Delta_2$ we say that $\Delta_1$ is a \emph{subenvironment} of $\Delta_2$.
\end{definition}

Note that the notion of subenvironment only concerns the datatypes that are defined in an environment and not the definitions of those datatypes (the constructors and their types), which will be treated by the notion of \emph{refinement} in the sequel.

\section{Datatype Refinement}\label{sec:refinement}

Henceforth we will fix a particular datatype environment $\underline{\Delta} : \ul{D} \to \Map\ \Con\ (\Sch\ \ul{D})$ which we call the \emph{underlying datatype environment}.  We think of $\underline{\Delta}$ as the datatype environment that is provided by the programmer by their datatype definitions.

\begin{example}\label{ex:datatype-env}
  We will use the following as running example of underlying datatype environment.
  Consider the datatype $\mathsf{Lam}$ of $\lambda$-terms with arithmetic using a locally nameless representation:
  \begin{lstlisting}
data Arith = Lit Int | Add | Mul 
data Lam = Cst Arith | BVr Int | FVr String | Abs Lam | App Lam Lam
  \end{lstlisting}
  These datatypes are slightly artificial, but they allow us to illustrate several features of the definitions in one example.  For simplicity, we will consider $\mathsf{Int}$ and $\mathsf{String}$ to be base types (which will, therefore, not be refined).
\end{example}

The underlying datatype environment contains definitions for all the datatypes declared by the programmer.
Some datatype definitions require the definitions of other datatypes to be understood properly. 
For example, to understand \textsf{Lam}, one must also understand the definition of \textsf{Arith} since one is defined in terms of the other.
There is a notion of a subenvironment that contains all and only those definitions that are needed to understand one particular datatype.

\begin{definition}[Slice]\label{def:gen-subenv}
  Suppose $\Delta : D \to \Map\ \Con\ (\Sch\ D)$.
One can always construct a subenvironment of $\Delta$ by starting from a given datatype $d \in D$ and closing under transitive dependencies.
  Define the \emph{slice of $d$ through $\Delta$}, written $\subenv{d}{\Delta}$, as the least subenvironment of $\Delta$ containing $d$.
\end{definition}

For example, in the environment $\ul{\Delta}$ described in Example \ref{ex:datatype-env}, we have $\subenv{\mathsf{Lam}}{\ul{\Delta}}$ being the whole environment, and $\subenv{\mathsf{Arith}}{\ul{\Delta}}$ being just the definition of $\mathsf{Arith}$ itself.

\begin{definition}[Refinement]
  We say that a datatype environment $\Delta_1 : D \to \Map\ \Con\ (\Sch\ D)$ is a \emph{refinement} of a datatype environment $\Delta_2 : D \to \Map\ \Con\ (\Sch\ D)$ just if the definitions of the datatypes in $\Delta_1$ are bounded above by the definitions in $\Delta_2$, that is: for all $d \in D$, $\Delta_1(d) \subseteq \Delta_2(d)$ (i.e. the graph of the first map is included in the graph of the second).
\end{definition}

Suppose $\Delta_1(d)(k)$ is some type scheme $S$ and $\Delta_1$ is a refinement of $\Delta_2$, then $\Delta_2(d)(k)$ is the same $S$.  So, the refinements of $\Delta : D \to \Map\ \Con\ (\Sch\ D)$ are in one-to-one correspondence with choices of constructors for each of the constituent datatypes.
More precisely, every function $f : \Pi d \in D.\ \mathcal{P}(\dom\ \Delta(d))$ determines a refinement $\Delta_f$ satisfying: 
\[
  \Delta_f(d)(k) = 
    \begin{cases} 
      \Delta(d)(k) & \text{if $k \in f(d)$} \\ 
      \bot & \text{otherwise} 
    \end{cases}
\]
and each refinement arises in this way.

\begin{example}\label{ex:refinement}
  The following refinement of the underlying environment from Example~\ref{ex:datatype-env} describes a type of closed, applicative terms over linear arithmetic.
  \begin{lstlisting}
      data Arith = Lit Int | Add Arith
      data Lam = Cst Arith | App Lam Lam
  \end{lstlisting}
  This refinement is determined by the choice $f_{LL}$ satisfying: 
  \[
  f_{LL}(\mathsf{Arith}) = \{\mathsf{Lit},\,\mathsf{Add}\} \qquad
  f_{LL}(\mathsf{Lam}) = \{\mathsf{Cst},\,\mathsf{App}\}
  \]
\end{example}

For the purpose of assigning types to the program, we construct a new datatype environment consisting of all possible refinements of the underlying environment supplied by the programmer\footnote{It would suffice to take all the refinements of all the slices (which itself still includes some redundancy, but this would complicate the definitions for no practical gain)}.

\begin{definition}[The Intensional Refinement Environment]
  Given a family of datatype environments $\Delta_{i \in I} : D_i \to \Map\ \Con\ (\Sch\ D_i)$ we define their coproduct as the environment:
\[
  \coprod_{i \in I} \Delta_i : (\coprod_{i \in I} D_i) \to \Map\ \Con\ (\Sch\ (\coprod_{i \in I} D_i))
\]
whose domain is simply a disjoint sum of sets (as defined in the preliminaries).
The coproduct comes equipped with canonical injections $\inj_i$ satisfying $(\coprod_{i \in I} \Delta_i)(\inj_i d)(d)(k) = \inj_i(\Delta_i(d)(k))$ wherever the latter is properly defined.

The \emph{intensional refinement environment}, written $\canon{\Delta}$, is the coproduct of all the refinements of $\ul{\Delta}$:
\[
  \canon{\Delta} \coloneqq \coprod_{f \in I} \ul{\Delta}_f : \canon{D} \to \Map\ \Con\ (\Sch\ \canon{D})
\]
where $I$ is the set $\Pi\, \ul{d} \in \ul{D}.\ \mathcal{P}(\dom\ \ul{\Delta}(\ul{d}))$ of functions from underlying datatypes $d$ to appropriate subsets of constructors.
That is, the set of all possible refinements.
We write the domain of this coproduct as $\canon{D}$.
\end{definition}

Note that, formally, the datatype identifiers whose definitions are given in the intensional refinement environment are of shape $\inj_\Delta\,d$ with $d \in \mathbb{D}$.
This is convenient because one can read such a name as ``the refinement of $d$ whose definition is $\Delta$''.
However, we will continue to use more friendly names (and Haskell notation) in our examples.

\begin{example}\label{ex:lam-refinements}
  The type of closed, applicative terms over linear arithmetic from Example~\ref{ex:refinement} can be found in $\canon{\Delta}$ alongside the type $\mathsf{LArith}$ of linear arithmetic constants:
  \[
    \mathsf{LATm_0} \coloneqq \inj_{f_{LL}}\,\mathsf{Lam} \qquad
    \mathsf{LArith} \coloneqq \inj_{f_{LA}}\,\mathsf{Arith}
  \]
  the latter being defined by $f_{LA}(\mathsf{Arith}) = \{\mathsf{Lit},\,\mathsf{Add}\}$ and $f_{LA}(\mathsf{Lam}) = \emptyset$. This datatype could equivalently be defined as $\inj_{f_{LL}}\,\mathsf{Arith}$ (or in many other ways).
  Other refinement datatypes defined include the type $\mathsf{LLam_0}$ of closed $\lambda$-terms over linear arithmetic, the type $\mathsf{ATm}$ of applicative terms and the type $\mathsf{ATm_0}$ of closed applicative terms, whose definitions we give in convenient notation:
  \begin{lstlisting}
      data ATm$_0$ = Cst Arith | App ATm$_0$ ATm$_0$ 
      data ATm = FVr String | Cst Arith | App ATm ATm
      data LLam$_0$ = Cst LArith | BVr Int | Abs LLam$_0$ | App LLam$_0$ LLam$_0$
  \end{lstlisting}
\end{example}

\begin{definition}[Refinement Type]
  A type (scheme) $S \in \Sch\ \canon{D}$ is said to be a \emph{refinement type}.  Refinement types come equipped with an \emph{underlying type}, written $\under(S)$, which is a type (scheme) in $\Sch\ \ul{D}$ defined recursively as follows:
  \[
    \begin{array}{rcl}
      \under(a) &=& a \\
      \under(b) &=& b \\
      \under(\inj_f d\ \vv{T}) &=& d\ \vv{\under(T)} \\
      \under(T_1 \to T_2) &=& \under(T_1) \to \under(T_2) \\
      \under(\forall a.\ S) &=& \forall a.\ \under(S)
    \end{array}
  \]
\end{definition}



In the following, we will assume that we are given a program equipped with a complete underlying typing, that is: every subterm $M$ has an associated type $S \in \Sch\ \ul{D}$.
Our task will be to find a new \emph{refinement} typing: an assignment to each subterm $M$ of a refinement type $S' \in \Sch\ \canon{D}$ that has the same ``shape'' as the underlying type $S$ of $M$ in the sense that $\under(S') = S$.

\section{Subtyping}\label{sec:subtyping}

\begin{figure*}
\[
  \begin{array}{c}
    \prftree[l,r]
  {\small$\left|\begin{array}{l}\under(T_1) \neq \under(T_2)\end{array}\right.$}
    {\rlnm{SShape}}
    {
      \types T_1 \not\subtype T_2
    }
    \\[8mm]
    \prftree[l,r]
  {\small$\left|\begin{array}{l}\dom(\canon{\Delta}(d_1)) \not\subseteq \dom(\canon{\Delta}(d_2))\end{array}\right.$}
    {\rlnm{SMis}}
    {  \types d_1\ \vv{T_1} \not\subtype d_2\ \vv{T_2} }
    \\[8mm]
    \prftree[l,r]
    {\small$\left|
      \begin{array}{l}
        m = \arity(k),\,
        i \in [1..m]\\[1mm]
        \canon{\Delta}(d_1)(k) = \forall \vv{\alpha}.\ U_{1_1} \to \cdots U_{1_m} \to d_1\ \vv{\alpha}\\[1mm]
        \canon{\Delta}(d_2)(k) = \forall \vv{\beta}.\ U_{2_1} \to \cdots U_{2_m} \to d_2\ \vv{\beta}
      \end{array}
      \right.
    $}
    {\rlnm{SSim}}
    {
      \types U_{1_i}[\vv{T_1/\alpha}] \not\subtype U_{2_i}[\vv{T_2/\beta}]
    }
    {
      \types d_1\ \vv{T_1} \not\subtype d_2\ \vv{T_2}
    }
    \\[8mm]
    \prftree[l]{\rlnm{SArrL}}
    {
      \types T_1' \not\subtype T_1
    }
    {
      \types T_1 \to T_2 \not\subtype T_1' \to T_2'
    }
    \qquad\qquad
    \prftree[l]{\rlnm{SArrR}}
    {
      \types T_2 \not\subtype T_2'
    }
    {
      \types T_1 \to T_2 \not\subtype T_1' \to T_2'
    }
  \end{array}
\]
\caption{The complement of the subtyping relation on monotypes.}\label{fig:subtype-simple}
\end{figure*}

Refinement induces a natural ordering on refinement datatypes according to which constructors are available in their definition.
This ordering can then be lifted to all types built over those datatypes in the obvious way.

\begin{definition}[Subtyping]
The judgement $\types T_1 \subtype T_2$ is defined coinductively, using the inductive system of rules in Figure \ref{fig:subtype-simple} to characterise the complement.
We extend subtyping to two schemes that have the same quantifier prefix, writing:
  $\types \forall \vv{\alpha}.\ T_1 \subtype \forall \vv{\alpha}.\ T_2$
whenever $\types T_1 \subtype T_2$.
We say that two types $T_1$ and $T_2$ are \emph{subtype equivalent} and write $\types T_1 \equiv T_2$ just if $\types T_1 \subtype T_2$ and $\types T_2 \subtype T_2$.
\end{definition}

\begin{toappendix}
It is straightforward to show the following by coinduction:
\begin{lemma}
  Subtyping is a preorder.
\end{lemma}
\begin{proof}
  The proof of reflexivity is an easy coinduction.
  We prove that, if there is $T$ such that $T_1 \subtype T$ and $T \subtype T_2$ then $T_1 \subtype T_2$ by coinduction on $\subtype$.
  This amounts to showing that the set $\{(T_1,\,T_2) \mid \exists T.\ T_1 \subtype T \subtype T_2\}$ is closed under the given rules, so in each case, we will show that  existence of such an intermediate $T$ is in contradiction with the given premises in the case.
  \begin{description}
    \item[\rlnm{SShape}]
      Suppose there is $T$ such that $T_1 \subtype T$ and $T \subtype T_2$ and $r(T_1) \neq r(T_2)$.
      Since $\subtype$ satisfies \rlnm{SShape}, it follows that $r(T_1) = r(T)$ and $r(T) = r(T_2)$ from which we obtain the desired contradiction.
    \item[\rlnm{SMis}]
      Suppose there is $T$ such that $d_1\ \vv{T_1} \subtype T \subtype d_2\ \vv{T_2}$.
      It follows from \rlnm{SShape} that $T$ is some datatype $d_3\ \vv{T_3}$.
      Furthermore, suppose $\dom(\canon{\Delta}(d_1)) \not\subseteq \dom(\canon{\Delta}(d_2))$.
      By \rlnm{SMis}, $\dom(\canon{\Delta}(d_1)) \subseteq \dom(\canon{\Delta}(d_3)$, and $\dom(\canon{\Delta}(d_3)) \subseteq \dom(\canon{\Delta}(d_2)$.
      Hence by the transitivity of set inclusion, we immediately have a contradiction.

    \item[\rlnm{SSim}]
      Again, suppose there is $T$ such that $d_1\ \vv{T_1} \subtype T \subtype d_2\ \vv{T_2}$ and thus $T$ is some datatype $d_3\ \vv{T_3}$.
      Suppose $k \in \dom(\canon{\Delta}(d_1)$, by \rlnm{SMis} it is also defined for $d_2$ and $d_3$.
      Suppose $U_{1_i}[\vv{T_1/\alpha_1}]$, $U_{2_i}[\vv{T_2/\alpha_2}]$, $U_{3_i}[\vv{T_3/\alpha_3}]$ are the ith arguments, instantiated with the corresponding type argument, to the constructor in each datatype respectively.  Suppose $U_{1_i}[\vv{T_1/\alpha_1}] \not\subtype U_{3_i}[\vv{T_3/\alpha_3}]$ so there is no intermediate type $T_i$ such that $U_{1_i}[\vv{T_1/\alpha_1}] \subtype T_i \subtype U_{3_i}[\vv{T_3/\alpha_3}]$.  However, it follows from our original assumption and \rlnm{SSim} that $U_{1_i}[\vv{T_1/\alpha_1}] \subtype U_{2_i}[\vv{T_2/\alpha_2}]$ and similarly that $U_{2_i}[\vv{T_2/\alpha_2}] \subtype U_{3_i}[\vv{T_3/\alpha_3}]$.
      Thus we reach a contradiction.

    \item[\rlnm{SArrL}]
      Suppose there is $T$ such that $T_1 \to T_2 \subtype T \subtype T_1' \to T_2'$ and suppose there is no intermediate $U$ such that $T_1' \subtype U \subtype T_1$.
      It follows from \rlnm{SShape} that $T$ must be of shape $T_3 \to T_4$.
      It follows from \rlnm{SArrL} that, therefore, $T_1' \subtype T_3$ and $T_3 \subtype T_1$, in contradiction of the absence of $U$.
    \item[\rlnm{SArrR}]
      Follows analogously to the above case.
  \end{description}
\end{proof}
\end{toappendix}

Intuitively, refinement specifies the possible shapes of types that are then interrelated by subtyping.
Refinement is a covariant treatment of arrow types, since we have $\under(T_1 \to T_2) = T_1' \to T_2'$ iff $\under(T_1) = T_1'$ and $\under(T_2) = T_2'$.
On the other hand, as can be seen from the definition, subtyping interprets the argument type contravariantly.
Consequently, there are some $\under(T) = \ul{T}$ for which $T \subtype \ul{T}$ and other $\under(T) = \ul{T}$ for which $\ul{T} \subtype T$ (viewing an underlying type as its own trivial refinement).

We give the definition coinductively because, as usual, there is a notion of simulation that arises naturally from our coalgebraic view of datatype environments.  
Consequently, it is most straightforward to think of the defining rules as providing a system in which to construct finite refutations of subtype inequalities $T_1 \not\subtype T_2$, which will, ultimately, fail to hold either because the types $T_1$ and $T_2$ have a different shape, or because $T_1$ provides some constructor that $T_2$ does not.

\begin{example}\label{ex:atms}
  Following the running example, the judgement $\types \mathsf{LATm}_0 \to \mathsf{String} \not\subtype \mathsf{ATm}_0 \to \mathsf{String}$ follows by a simple refutation:
  \[
    \prftree[l]{\rlnm{SArrL}}
    {
      \prftree[l]{\rlnm{SSim}}
      {
        \prftree[l]{\rlnm{SMis}}
        {
          \types \mathsf{Arith} \not\subtype \mathsf{LArith}   
        }
      }
      {
          \types \mathsf{ATm}_0 \not\subtype \mathsf{LATm}_0
      }
    }
    {
        \types \mathsf{LATm}_0 \to \mathsf{String} \not\subtype \mathsf{ATm}_0 \to \mathsf{String}
    }
  \]
\end{example}

Conversely, we can use the coinduction principle to show that $\types T_1 \subtype T_2$, in which case we require a model\footnote{By which we mean a set of pairs of types satisfying all $\not\subtype$-defining rules.} of $\subtype$ that contains $(T_1,\,T_2)$.  

\begin{example}\label{ex:subtype-witness}
  For $\types \mathsf{ATm}_0 \to \mathsf{String} \subtype \mathsf{LATm}_0 \to \mathsf{String}$ we provide the following witness:
  \[
    \left\{
    \begin{array}{c}
    (\mathsf{ATm}_0 \to \mathsf{String}, \mathsf{LATm}_0 \to \mathsf{String}),\,
    (\mathsf{String}, \mathsf{String}),\\
    (\mathsf{LATm}_0, \mathsf{ATm}_0),\,
    (\mathsf{LArith}, \mathsf{Arith}),\,
    (\mathsf{Int}, \mathsf{Int})
    \end{array}
    \right\}
  \]
  It can be easily verified that this set is a model of the defining rules for $\subtype$ and hence, by coinduction, is contained within it.
\end{example}

However, such models can be a bit unwieldy in general as the types involved get more complex.  We can do better by observing that the definition can be approximated by a coinductive part, concerning datatypes, and an inductive part, by which a subtyping relationship between datatypes is lifted to all types.  Consequently, we need only find a model of the coinductive part, which is much neater since it only concerns $\Dt\ D \times \Dt\ D$.  The following can be shown by a straightforward coinduction.

\begin{lemmarep}[Simulation]\label{lem:simulation}
  Let $R \subseteq \Dt\ D \times \Dt\ D$ and suppose that, for all $(d_1\ \vv{T_1} ,\,d_2\ \vv{T_2}) \in R$, and for all $k$ such that $\Delta(d_1)(k)$ is defined:
  \begin{itemize}
    \item $\Delta(d_2)(k)$ is defined.
    \item And, moreover, $\Ty(R)(U_{1_i} ,\,U_{2_i})$ for each $i \in [1..\arity(k)]$, where $\vv{U_1}$ and $\vv{U_2}$ are the argument types of $\Delta(d_1)(k)$ and $\Delta(d_2)(k)$ instantiated at $\vv{T_1}$ and $\vv{T_2}$ respectively.
  \end{itemize}
  Then it follows that $\Ty(R)$ is included in the subtype relation.
\end{lemmarep}
\begin{proof}
  The proof is by coinduction.
  \begin{description}
    \item[\rlnm{SShape}] By definition, if $\Ty(R)(T_1,\,T_2)$ then $T_1$ and $T_2$ have the same shape.
    \item[\rlnm{SMis}] Suppose $\dom(\canon{\Delta}(d_1)) \not\subseteq \dom(\canon{\Delta}(d_2))$ and suppose, for the purpose of obtaining a contradiction, that $\Ty(R)(d_1\,\vv{T_1},\,d_2\,\vv{T_2})$.  Then there is some $k$ such that $\canon{\Delta}(d_1)(k)$ is defined, but $\canon{\Delta}(d_2)(k)$ is not, thus contradicting the first bullet of the definition of $R$.
    \item[\rlnm{SData}] Suppose $(U_{1_i}[\vv{T_1/\alpha}],\,U_{2_i}[\vv{T_2/\alpha}]) \notin \Ty(R)$ and the side conditions on the rule hold.  Then suppose for contradiction that $\Ty(R)(d_1\,\vv{T_1},\,d_2\,\vv{T_2})$.  By definition, it must be that $R(d_1\,\vv{T_1},\,d_2\,\vv{T_2})$, thus contradicting the second bullet.
    \item[\rlnm{SArrL}] Suppose $\Ty(R)(T_1 \to T_2,\,T_1' \to T_2')$.  Then, by definition of $\Ty(R)$ it can only be because $\Ty(R)(T_1',\,T_1)$.
    \item[\rlnm{SArrR}] Analogously.
  \end{description}
\end{proof}

\noindent
Using this result, it suffices to exhibit $R \coloneqq \{(\mathsf{LATm}_0, \mathsf{ATm}_0),\,(\mathsf{LArith}, \mathsf{Arith})\}$ in order to conclude e.g. $\types \mathsf{ATm}_0 \to \mathsf{String} \subtype \mathsf{LATm}_0 \to \mathsf{String}$.  Intuitively, this witness determines the model $\Ty(R)$, which contains the model of Example \ref{ex:subtype-witness}.


\section{Refinement Type Assignment}\label{sec:type-assignment}

In this section, we present a refinement type system whose purpose is to exclude the possibility of pattern-match failure.  To achieve this, the typing rule for pattern-matching requires that cases are exhaustive according to the type of the scrutinised expression.  
However, the system allows for all refinement datatypes and incorporates the above notion of subtyping, which allows for the scrutinised expression to be typed much more precisely than is possible in the underlying type system.

For the purpose of defining the refinement type system, we make some standard Hindley-Damas-Milner assumptions about the underlying type system, namely that type application happens immediately after introducing a variable of polymorphic type and type abstraction happens only at the point of definition.  As a minor simplification, we assume that constants are monomorphic and write $\mathbb{C}(c)$ for the monotype assigned to $c$ axiomatically.
Since this is a refinement type system, we assume that all expressions have already been assigned an underlying type, which we will typically write with an underline to aid readability.  
We will only need to consult these underlying types when they appear in the syntax, e.g. abstraction and type application.

Additionally, we relax the normal definition of a type environment from a function to a relation.
Program variables may, therefore, have many types as long as they refine the same underlying type.
This assumption is equivalent to allowing environment-level intersection types.

\begin{definition}[Type assignment]\label{def:typing}
A \emph{type environment}, typically $\Gamma$ or $\Delta$, is a finite relation between program variables $x$ and type schemes $S$, whose elements are typically written $x:S$.
We require that $x:S_1 \in \Gamma$ and $x:S_2 \in \Gamma$ implies $\under(S_1) = \under(S_2)$.
This ensures that $\under(\Gamma)$ can be defined in the obvious way.
The type assignment system is divided into two sets of rules, for expressions (Figure~\ref{fig:type-system}) and for modules (Figure~\ref{fig:typing-mod}), defining judgements, respectively:
\[
  \Gamma \types e : T \qquad \Gamma \types m : \Delta
\]
in which $\under(\Gamma) \types e : \under(T)$ and $\under(\Gamma) \types m : \under(\Delta)$ are the underlying typings, provided by the programming language, for the expression $e$ and the module $m$ respectively.
\begin{figure}
\[
  \begin{array}{c}
  \prftree[l]{\rlnm{TModE}}{\Gamma \types \epsilon : \Gamma}
  \qquad
    \prftree[l]{\rlnm{TModD}}{\Gamma \types m : \Gamma'}{\Gamma' \cup \{x : T\} \types e : T'}{\types T' \subtype T}{\Gamma \types m \cdot \langle x = \Abs{\vv{\alpha}}{e} \rangle : \Gamma' \cup \{x : \forall \vv{\alpha}.\ T\}}
  \end{array}
\]
\caption{Typing for modules.}\label{fig:typing-mod}
\end{figure}
\end{definition}

\begin{toappendix}
\begin{lemma}[Weakening]
  Suppose $\types \Gamma_2 \subtype \Gamma_1$ and $\types T_1 \subtype T_2$. If $\Gamma_1 \types e : T_1$ then $\Gamma_2 \types e : T_2$.
\end{lemma}
\begin{proof}
  We prove that, for all $\Gamma'$, $\Gamma' \subtype \Gamma$ implies $\Gamma' \types e : T$, by induction on $\Gamma \types e : T$.  Then, if also $T \subtype T'$, $\Gamma' \types e : T'$ follows by \rlnm{TSub}.
  \begin{description}
    \item[\rlnm{TVar}] Suppose $\vv{T} :: \vv{\underline{T}}$ and $x : \forall \vv{\alpha}.\ U \in \Gamma$ and let $\Gamma'$ be such that $\Gamma' \subtype \Gamma$.  Then $x : \forall \vv{\alpha}.\ V \in \Gamma'$ with $V \subtype U$. By \rlnm{TVar}, $\Gamma' \types x [\vv{\underline{T}}] : V[\vv{T}/\vv{\alpha}]$.  It follows by definition of subtyping that $V[\vv{T}/\vv{\alpha}] \subtype U[\vv{T}/\vv{\alpha}]$ and therefore the desired result follows by \rlnm{TSub}.
    \item[\rlnm{TSub}\rlnm{TCst}\rlnm{TCon}\rlnm{TApp}] In these cases, the conclusion follows from the hypotheses independently of the environment. 
    \item[\rlnm{TAbs}] Suppose $T_1 :: \underline{T_1}$ and $x \notin \dom(\Gamma)$ and then suppose that $\Gamma'$ is such that $\Gamma' \subtype \Gamma$. Then, by definition and reflexivity of subtyping, also $\Gamma' \cup \{x:T_1\} \subtype \Gamma \cup \{x:T_1\}$.  It follows from the induction hypothesis, therefore, that $\Gamma' \cup \{x:T_1\} \types e : T_2$.  We may assume that $x \notin \dom(\Gamma')$ by the variable convention.  Therefore, the result follows by \rlnm{TAbs}.
    \item[\rlnm{TCase}] Suppose $\dom(\Delta(d)) \subseteq \{k_1,\ldots,k_m\}$.  Suppose $\Gamma' \subtype \Gamma$.  It follows immediately from the induction hypothesis that $\Gamma' \types e : d$.  It follows by reflexivity of subtyping that, for each $i$, $\Gamma' \cup \vv{x:\Delta(d)(k_i)} \subtype \Gamma \cup \vv{x:\Delta(d)(k_i)}$.  Hence, the induction hypothesis gives, for each $i$, $\Gamma' \cup \vv{x:\Delta(d)(k_i)} \types e_i : T$.  The result follows immediately by \rlnm{TCase}.
  \end{description}
\end{proof}
\end{toappendix}

The system is conceptually similar to an underlying ML-style system, but note:
\begin{itemize}
  \item Any suitable refinement datatype $d$ can be used in order to type a datatype constructor or the scrutinee of a case statement.
  \item The notion of subtyping from the previous section is incorporated through a subsumption rule (recall that $\types T_1 \subtype T_2$ implies that $T_1$ and $T_2$ have the same shape according to $\under$).
  \item The pattern-matching rule is restricted by a condition requiring that cases are exhaustive.
  \item The branches of the case expression only need to be typed if the branch is reachable, incorporating path-sensitivity.  This relaxation only makes sense for a refinement type system, because reachability is encoded by choosing an appropriate refinement $d$ in the rule \rlnm{TCase}.  From an operational point of view it makes no difference to the set of computations expressible.
  \item Finally, everywhere a particular underlying type is required by the syntax, an arbitrary choice of refinement type of the appropriate shape can be made in its place.
\end{itemize}

As discussed in the introduction, allowing several types for each term ensures they can be used in different contexts.
This approach is more lightweight than an intersection type system, and arguably easier for programmers to  reason about if types are to be considered as certificates.
When it comes to algorithmic inference, however, the non-deterministic aspect would be problematic.
Instead, in Section \ref{sec:inference}, we rely on \emph{refinement polymorphism} to summarise \emph{every} typing of a variable in some environment compactly by a single constrained type scheme.
The polymorphism of this kind is no different from that of the Hindley-Milner system, which could equally be viewed as an infinite intersection type system, or indeed allowing several typings of the same variable in an environment.
Likewise, it is simpler to define polymorphic constructors and datatypes, than to consider each instantiation separately.

\begin{figure*}
\[
  \begin{array}{c}

  \prftree[r,l]{\small$\left|\begin{array}{l}
    \under(\vv{T}) = \vv{\underline{T}}\\[1mm]
    x : \forall \vv{\alpha}.\ T \in \Gamma
  \end{array}\right.$}
  {\rlnm{TVar}}
  {
    \Gamma \types x\ \vv{\underline{T}} : T[\vv{S/\alpha}]
  }
    \qquad\quad
  \prftree[l]{\rlnm{TCst}}{\Gamma \types c : \mathbb{C}(c)}
    \\[8mm]
  \prftree[r,l]
    {
      \small$\left|
        \begin{array}{l}
    \under(\vv{T}) = \vv{\underline{T}}\\[1mm]
          k \in \dom\,(\canon{\Delta}(d)) \\[1mm]
          \canon{\Delta}(d)(k) = \forall \vv{\alpha}.\ T
        \end{array}
        \right.$
    }
    {\rlnm{TCon}}
    {\Gamma \types k\ \vv{\underline{T}} : T[\vv{S/\alpha}]}
    \qquad\quad
    \prftree[l,r]{$\left|\begin{array}{l}\types T_1 \subtype T_2\end{array}\right.$}{\rlnm{TSub}}{\Gamma \types e : T_1}{\Gamma \types e : T_2}
    \\[8mm]
    \prftree[r,l]{\small$\left|\begin{array}{l}\under(T_1) = \underline{T_1}\\[1mm]x \notin \dom\,\Gamma\end{array}\right.$}{\rlnm{TAbs}}{\Gamma \cup \{x:T_1\} \types e : T_2}{\Gamma \types \abs{x\!:\!\underline{T_1}}{e} : T_1 \to T_2}
    \qquad\quad
    \prftree[l]{\rlnm{TApp}}{\Gamma \types e_1 : T_1 \to T_2}{\Gamma \types e_2 : T_1}{\Gamma \types e_1\ e_2 : T_2}
    \\[8mm]
    \prftree[r,l]{\small$\left|\begin{array}{l}
      \dom(\canon{\Delta}(d)) = \{k_1,\ldots,k_m\}
    \end{array}\right.$}
    {\rlnm{TCase}}
    {
      \prfassumption{\Gamma \types e : d\,\vv{T}}
    }
    {
      \prfassumption{(\forall i \leq m)\: \Gamma \cup \{ \vv{ x_i : \canon{\Delta}(d)(k_i)[\vv{T/\alpha}] } \}\ \types e_i : T}
    }
    {
      \Gamma \types \matchtm{e}{\mathord{\mid}_{i=1}^m\,k_i\ \vv{x_i} \mapsto e_i} : T
    }\\[4mm]
  \end{array}
\]
\caption{Type assignment for expressions.}\label{fig:type-system}
\end{figure*}

\vspace{3cm}
\begin{example}
  Recall the refinements of Example~\ref{ex:lam-refinements} and consider the function $\mathsf{cloSub}$, with underlying type $\mathsf{List}\ (\mathsf{String} \times \mathsf{Lam}) \to \mathsf{Lam} \to \mathsf{Lam}$, whose purpose is to close an applicative term by substituting closed terms everywhere.
  \begin{lstlisting} 
    cloSub m t = 
      case t of
        FVr s   -> lkup m s
        Cst c   -> Cst c
        App u v -> App (cloSub m u) (cloSub m v)
  \end{lstlisting}
  To keep the example simple, we assume that the lookup function $\mathsf{lkup}$ has the following type: $\forall \alpha.\ \mathsf{List}\ (\mathsf{String} \times \alpha) \to \mathsf{String} \to \alpha$ in the environment. 
  Consequently, in our official syntax, e.g. the $\mathsf{FVr}$ case really contains an explicit type application: $\mathsf{lkup}\ \mathsf{Lam}\ m\ s$.
  Then the function $\mathsf{cloSub}$ can be assigned the refinement type\footnote{Here, \textsf{List} and $\times$ can be understood as the trivial refinements of their namesakes, i.e. with all constructors available.}: $\mathsf{List}\ (\mathsf{String} \times \mathsf{ATm}_0) \to \mathsf{ATm} \to \mathsf{ATm}_0$.
  Thus expressing the fact that the application of a \emph{closing} substitution to a arbitrary \emph{applicative} term yields a \emph{closed applicative} term.
  
  This is possible due to a combination of the features of the system.
  First, observe that it is possible, in the abstraction rule, to assume that the bound variable $m$ of underlying type $\mathsf{List}\ (\mathsf{String} \times \mathsf{Lam})$ has type $\mathsf{List}\ (\mathsf{String} \times \mathsf{ATm}_0)$ in \rlnm{TAbs} since it can easily be seen that the former is a refinement of the latter.
  Then it follows that $\mathsf{lkup}\ \mathsf{Lam}\ m\ i$ can be assigned the type $\mathsf{Atm}_0$ by choosing $\mathsf{Atm}_0$ for $T$ in \rlnm{TVar}.
  Second, under the assumption that the bound variable $t$ has refinement type $\mathsf{Atm}$, it follows from \rlnm{TCase} that the variable $c$ that is bound by the case $\mathsf{Cst}\ c$ can be assigned the type $\mathsf{Arith}$.
  Note that the rule \rlnm{TCase} is applicable only because we have chosen the refinement $\mathsf{ATm}$ of $\mathsf{Lam}$ which guarantees that the input will not contain any abstractions.
  Then, in the body of the case, we can choose instead the more specific typing $\mathsf{Cst} : \mathsf{Arith} \to \mathsf{ATm}_0$.
  Similarly, $u$ and $v$ are assigned the type $\mathsf{ATm}$ so that the subexpressions $\mathsf{cloSub}\ m\ u$ and $\mathsf{cloSub}\ m\ v$ in the body of the final case can be assigned the type $\mathsf{ATm}_0$.
  Then the type of the body as a whole, and therefore the entire case analysis, is also $\mathsf{Atm}_0$.
\end{example}

The central problem is typability, for closed expressions: given an underlying datatype environment $\ul{\Delta}$ and a closed module $m$ which is typed in $\ul{\Delta}$, does there exist a refinement type assignment to the functions of $m$?
Typically $m$ will contain library functions whose source is not available to the system, but for which an underlying type is known.
To incorporate such functions we interpret an underlying type environment $\ul{\Gamma}$ as containing trivial refinement types for each such function, i.e. each $d$ occurring in such a type $\Gamma$ denotes the refinement of $d$ that makes available all constructors.

\begin{definition}[Typability]\label{def:typability}
  A triple $\ul{\Delta}$, $\Gamma$ and $m$ constitutes a positive instance of the \emph{refinement typability problem} just if there is a refinement type environment $\Gamma'$ such that $\Gamma \types m : \Gamma'$.  In such a case, we say that $\ul{\Delta};\,\Gamma \types m$ is \emph{refinement typable}.
\end{definition}

\noindent
The rest of the paper concerns the algorithmic solution of the typability problem.



\section{Constructor Set Constraints}\label{sec:constraints}

We assume a countable set of \emph{refinement variables}, ranged over by $X$, $Y$, $Z$ and so on.
The purpose of a refinement variable $X$ is to represent a function in $\Pi d \in \ul{D}.\ \mathcal{P}(\dom\,\ul{\Delta}(d))$.  As described in Section~\ref{sec:refinement}, such functions are in 1-1 correspondence with refinements of $\ul{\Delta}$.
We will abuse notation and write $X$ for both uses (thus the following rather strange-looking equation $X(d) = \dom(X(d))$ holds by interpreting each of the two occurrences of $X$ according to its context.)
\begin{definition}[Constraints]
  A \emph{constructor set expression}, typically $S$, is either a finite set of constructors $\{k_1,\ldots,k_m\}$ or a pair $X(\ul{d})$ consisting of a refinement variable $X$ and an underlying datatype $\ul{d}$.  The underlying type of the constructor set expression is (partially) defined as follows:
  \[
    \under(X(\ul{d})) = \ul{d} \qquad\qquad
    \under(\{k_1,\ldots,k_m\}) = \ul{d} \quad\text{if $\forall i \in [1..m].\ k_i \in \dom\,\ul{\Delta}(\ul{d})$}
  \]
  We consider only those constructor set expressions for which the underlying type is defined.
  We write $\frv(S)$ for the set of refinement variables occurring anywhere in $S$ (which will either be empty or a singleton).

  An \emph{inclusion constraint} is an ordered pair of constructor set expressions, written (suggestively) as $S_1 \subseteq S_2$.  When $S_1$ is a singleton $\{k\}$, we will rather write the pair as $k \in S_2$.  We shall only consider inclusion constraints in which both set expressions have the same underlying type.  The  refinement variables of an inclusion constraint $\frv(S_1 \subseteq S_2)$ are defined by extension from $\frv(S_1)$ and $\frv(S_2)$ in the obvious way.

  A \emph{conditional constraint}, hereafter just \emph{constraint}, is a pair $\phi\ ?\ S_1 \subseteq S_2$ consisting of a set of inclusion constraints $\phi$ and an inclusion constraint $S_1 \subseteq S_2$.  The set $\phi$ is called the \emph{guard} and the inclusion $S_1 \subseteq S_2$ the \emph{body}.  We will only consider conditional constraints in which each element of the guard has shape $k \in X(\ul{d})$.  When the guard of a constraint $\emptyset\ ?\ S_1 \subseteq S_2$ is trivial, we shall usually omit it and write only the body $S_1 \subseteq S_2$.  The set of refinement variables $\frv(\phi\ ?\ \frv(S_1 \subseteq S_2))$ of a constraint  is defined as usual.

  Sometimes we shall guard a constraint set $C$, and write $\phi\ ?\ C$ for the set $\{\psi \cup \phi\ ?\ S_1 \subseteq S_2 \mid \psi\ ?\ S_1 \subseteq S_2\ \text{an element of}\ C\}$.  We write $\frv(C)$ for the set of refinement variables occurring in $C$.

\end{definition}

  Intuitively, an inclusion $S_1 \subseteq S_2$ is satisfied by any assignment to the refinement variables that makes $S_1$ included in $S_2$.  A constraint $\phi\ ?\ S_1 \subseteq S_2$ is satisfied if either some inclusion in the guard is not satisfied or the body is satisfied.

\begin{definition}[Satisfaction]
  A \emph{constructor set assignment}, hereafter just \emph{assignment}, is a total map $\theta$ taking each refinement variable $X$ to a constructor choice function from $\Pi d \in \ul{D}.\ \mathcal{P}(\dom\,\ul{\Delta}(d))$.
  The meaning of a constructor set expression $S$ under an assignment $\theta$ is a set of constructors $\theta\mng{S}$ defined as follows:
  \[
    \theta\mng{X(\ul{d})} = \theta(X)(\ul{d}) \qquad\qquad
    \theta\mng{\{k_1,\ldots,k_m\}} = \{k_1,\ldots,k_m\} 
  \]
  An inclusion constraint $S_1 \subseteq S_2$ is \emph{satisfied} by an assignment $\theta$, written $\theta \models S_1 \subseteq S_2$ just if $\theta\mng{S_1}$ is included in $\theta\mng{S_2}$.
  A constraint $\phi\ ?\ S_1 \subseteq S_2$ is \emph{satisfied} by an assignment $\theta$, written $\theta \models \phi\ ?\ S_1 \subseteq S_2$ just if, whenever $\theta \models k \in X(d)$ for every inclusion constraint $k \in X(d)$ in $\phi$, then $\theta \models S_1 \subseteq S_2$.
\end{definition}

\begin{definition}[Solutions]
  A \emph{solution} to a constraint set $C$ is an assignment $\theta$ satisfying every constraint in $C$, we write $\theta \models C$. We say that $C$ is \emph{solvable}, or \emph{satisfiable}, just if it has a solution.
\end{definition}

\begin{remark}\label{rmk:Horn}
  The full set constraint language is exactly the monadic class of first-order propositions~\cite{bachmair1993set}.
By applying the translation of that paper, it can be shown that guarded constraints of the form laid out above are (monadic) Horn clauses with constructors simply interpreted as constants.
\end{remark}



\section{Type Inference}\label{sec:inference}

\begin{figure}
    \[
      \begin{array}{c}
        \prftree[l]{\rlnm{ISBase}}
        {}
        {
          \types b \subtype b \infers \emptyset
        }
        \qquad\qquad
        \prftree[l]{\rlnm{ISTyVar}}
        {}
        {
          \types \alpha \subtype \alpha \infers \emptyset
        }
        \\[8mm]
        \prftree[l]{\rlnm{ISArr}}
        {
          \types T_{21} \subtype T_{11} \infers C_1
        }
        {
          \types T_{12} \subtype T_{22} \infers C_2
        }
        {
          \types T_{11} \to T_{12} \subtype T_{21} \to T_{22} \infers C_1 \cup C_2
        }
        \\[8mm]
        \prftree[l,r]{\small$\left|
      \begin{array}{l}
        \ul{\Delta}(\ul{d})(k) =\\
        \quad \forall \vv{\alpha}.\ U_1 \to \cdots U_n \to d\ \vv{\alpha}\\[1mm]
        C = \{ X(\ul{d}) \subseteq Y(\ul{d})\}\ \cup\\[1mm]
        \qquad \bigcup_{k} \bigcup_{i=1}^{n} (k \in X(\ul{d}))\ ?\ C_{k_i}
          \\[1mm]
      \end{array}\right.$}
        {\rlnm{ISData}}
          {(\forall k i.)\ \types \inj_X(U_i)[\vv{T_X/\alpha}] \subtype \inj_Y(U_i)[\vv{T_Y/\alpha}] \infers C_{k_i}}
        {
          \types \inj_X\,d\ \vv{T_X} \subtype \inj_Y\,d\ \vv{T_Y} \infers C}
      \end{array}
    \]
  \caption{Inference for subtype inequalities.}\label{fig:inf-sub}
\end{figure}

Since our system is effectively syntax directed (the subsumption rule can be factored into the other syntax-directed rules), type inference follows a standard pattern of constraint generation and satisfiability checking (see e.g. \cite{odersky-sulzmann-wehr-TSPOS1999}).
The constraints are subtype inequalities over refinement variables, but it is easily seen that, in our restricted setting, such inequalities are equivalent to conditional inclusion constraints between refinement variables and sets of datatype constructors.
To enable this approach, we extend the language of types so to allow datatypes parametrised by refinement variables.
\begin{definition}[Extended Types]
  The \emph{extended types} are monotypes extended with datatypes built over refinement variables:
\[
  T,\,U,\,V \Coloneqq \cdots \mid \inj_X\,d\ \vv{T}
\]

Note that the type arguments to an injected datatype identifier are also extended.
Expressions of the form $(\inj_X\ \mathsf{List}) (\inj_Z\ \mathsf{Int})$ are, therefore, well-formed.  Recall from Section~\ref{sec:refinement} that refinement datatype identifiers are of the form $\inj_\Delta\,d$, with $d$ an underlying datatype identifier, and should be thought of as specifying the refinement of $d$ whose datatype definition is given by $\Delta$.
The task of inference is to determine constraints on these $\Delta$ that enable a typing to be assigned and check that the constraints have a solution.

For convenience, we shall implicitly lift injections to any type, or sequence of types, written $\inj_X\,T$, so that the injection is distributed over datatypes in $T$.
In the context of extended types, we will associate a substitution action $\theta T$ with each constructor set assignment $\theta$ by lifting the definition $\theta(\inj_X\,d) \coloneqq \inj_{\theta(X)}\,d$ homomorphically over all extended types.  Finally, we write $\frv(T)$ for the set of refinement variables occurring in injections in $T$.
\end{definition}

We also adopt an extension of type schemes that are constrained:
\begin{definition}[Constrained Type Scheme]
  We subsume the type scheme $S$ by the \emph{constrained type scheme}, which has shape: $\forall \vv{\alpha}.\ \forall \vv{X}.\ C \implies T$, where $C$ is a constraint set and $T$ is an \emph{extended type}.  We define $\frv(\forall\vv{\alpha}.\ \forall \vv{X}.\ C \implies T) = (\frv(C) \cup \frv(T)) \setminus \vv{X}$.  A \emph{constrained type environment} is a finite mapping from program variables to constrained type schemes, whose elements are written $x:S$.  We define $\frv(\Gamma)$ in the obvious way.
\end{definition}
As is typical, there is generally no ``best'' monotype solution to a set of inclusion constraints, so constrained type schemes give us an internal representation for the set of all types assignable to a module-level function.  
For example, assuming constant combinator $\mathbf{K}$ defined as usual, it can be seen that the module-level recursive function $f = \abs{x\!:\!\mathsf{Lam}}{\mathbf{K}\ [\mathsf{Lam},\,\mathsf{Lam}]\ x\ (f\ (f\ x))}$ can be assigned the constrained type scheme : $\forall X Y.\ C \implies \inj_X\,\mathsf{Lam} \to \inj_Y\,\mathsf{Lam}$, with $C$:
\[
  \begin{array}{lcl}
  X(\mathsf{Lam}) \subseteq Y(\mathsf{Lam})&\phantom{woo}&
  Y(\mathsf{Lam}) \subseteq X(\mathsf{Lam})\\ 
  \mathsf{Cst} \in X(\mathsf{Lam})\ ?\ X(\mathsf{Arith}) \subseteq Y(\mathsf{Arith})&\phantom{woo}&
  \mathsf{Cst} \in Y(\mathsf{Lam})\ ?\ Y(\mathsf{Arith}) \subseteq X(\mathsf{Arith})
  \end{array}
\] 
Intuitively, its input flows to its output and conversely, so we require $\inj_X\,\mathsf{Lam} \subtype \inj_Y\,\mathsf{Lam}$ and $\inj_Y\,\mathsf{Lam} \subtype \inj_X\,\mathsf{Lam}$ (which is encoded by the above set constraints when we view the refinements $X$ and $Y$ as functions specifying the choice of constructors).  However, there is obviously no ``best'' instantiation of refinement variables $X$ and $Y$.

Constrained type environments can be understood as compact descriptions of ``ordinary'' type environments (in the sense of Definition \ref{def:typing}), which is made precise as follows.

\begin{definition}
  Define $\bananas{\Gamma}$ for the type environment that can be obtained from the \emph{closed} constrained type schemes in $\Gamma$, by instantiation of refinement quantifiers with every possible solution, that is, supposing $\Gamma$ is closed: $
    \bananas{\Gamma} \coloneqq \{\ x:\forall \vv{\alpha}.\ \theta T \mid  \theta \models C \wedge (x:\forall\vv{\alpha}.\ \forall \vv{X}.\ C \implies T) \in \Gamma\ \}
  $.
\end{definition}

Typical presentations of type inference by constraint generation involve choosing fresh type variables, which are then constrained.
Since we work with refinement types, it is more convenient to choose fresh refinement type templates, which are just refinement types that are everywhere parametrised by fresh refinement variables --- in the setting of refinement types, at the point at which inference would choose a fresh type, the underlying shape of the type is already known.
We write $\fresh(X)$ to assert that $X$ must be a fresh refinement variable (i.e. not already used in the current scope).
We extend the notion to fresh types $T$ of underlying shape $\ul{T}$.
\begin{definition}[Fresh Types]
  We write $\fresh_{\ul{T}}(T)$ for the following inductive predicate.
  \begin{itemize}
    \item For all $\alpha \in \mathbb{A}$, $\fresh_\alpha(\alpha)$.
    \item For all $b \in \mathbb{B}$, $\fresh_b(b)$.
    \item For $\ul{d} \in \ul{D}$, if $\fresh_{\ul{T}}(T)$ for every $T$ in $\vv{T}$, and $\fresh(X)$ then $\fresh_{\ul{d}}(\inj_X(\ul{d})\ \vv{T})$
    \item For all $T_1,T_2 \in \Ty\ \canon{D}$, $\ul{T_1},\ul{T_2} \in \Ty\ \ul{D}$, if $\fresh_{\ul{T_1}}(T_1)$ and $\fresh_{\ul{T_2}}(T_2)$ then $\fresh_{\tiny{\ul{T_1} \to \ul{T_2}}}(T_1 \to T_2)$.
  \end{itemize}
  The definition guarantees that $\under(T) = \ul{T}$.
  We extend the notion to sequences of types, writing $\fresh_{\vv{\ul{T}}}(\vv{T})$ to denote that the two sequences $\vv{T}$ and $\vv{\ul{T}}$ have the same length and are related pointwise by freshness.
\end{definition}

\begin{figure}
  \[
    \begin{array}{c}
    \prftree[l]{\rlnm{ICst}}
      {}
      {
        \Gamma \types c : \ul{T} \infers \ul{T},\,\emptyset
      }
      \\[8mm]
      \prftree[l,r]{\small$\left|\begin{array}{l}k \in \dom\,(\ul{\Delta}(\ul{d}))\\[1mm]\fresh(X)\ \textrm{and}\ \fresh_{\vv{\ul{T}}}(\vv{T})\end{array}\right.$}
      {\rlnm{ICon}}
      {
        \Gamma \types k\ \vv{\ul{T}} : \ul{V} \infers \inj_X\,\ul{V}\ \vv{T},\,\{ k \in X(d)\}
      }
      \\[8mm]
      \prftree[r,l]{
        \small$\left|\begin{array}{l}
          x:\forall \vv{\alpha}.\ \forall \vv{X}.\ C \implies U \in \Gamma \\[1mm]
          \fresh(\vv{Y})\ \textrm{and}\ \fresh_{\vv{\ul{T}}}(\vv{T})
        \end{array}\right.$}
      {\rlnm{IVar}}
      {}
      {
        \Gamma \types x\ \vv{\underline{T}} : \underline{V} \infers U[\vv{Y/X}][\vv{T/\alpha}],\,C[\vv{Y/X}]
      }
      \\[8mm]
      \prftree[r,l]{\small$\left|\begin{array}{l}\fresh_{\ul{T_1}}(T_1)\end{array}\right.$}{\rlnm{IAbs}}
      {
        \Gamma \cup \{x:T_1\} \types e : \ul{T_2} \infers T_2,\,C
      }
      {
        \Gamma \types \abs{x\!:\!\underline{T_1}}{e} : \ul{T_1} \to \ul{T_2} \infers T_1 \to T_2,\,C
      }
      \\[8mm]
      \prftree[l]{\rlnm{IApp}}
      {
        \Gamma \types e_1 : \ul{T_1} \to \ul{T_2} \infers T_1 \to T_2,\,C_1
      }
      {
        \Gamma \types e_2 : \ul{T_1} \infers T_3,\,C_2
      }
      {
        \types T_3 \subtype T_1 \infers C_3
      }
      {
        \Gamma \types e_1\ e_2 : \ul{T_2} \infers T_2,\,C_1 \cup C_2 \cup C_3
      }
      \\[8mm]
      \prflineextra=0em
      \prftree[l,r]{\!\small\begin{tabular}{c}$\left|\!
        \begin{array}{l}
          \fresh_{\ul{T}\cdots\ul{T}}(T\cdot\vv{T_i})\\[1mm]
          C =C_0 \cup\ \{X(\ul{d}) \subseteq \{k_1,\ldots,k_m\}\}\\[.5mm]\hspace{6mm}\cup \bigcup_{i=1}^m (k_i \in X(\ul{d})\ ?\ (C_i \cup C_i'))
            \\[1mm]
            \ul{\Delta}(\ul{d})(k_i) = \\[.5mm]
            \quad\forall \vv{\alpha}.\ A_1 \to \cdots A_n \to d\ \vv{\alpha}
        \end{array}\right.$\\[2mm]\phantom{hello}\\[2mm]\hphantom{hi}\end{tabular}}
          {\rlnm{ICase}}
      {
        \prfassumption{\begin{array}{c}
          (\forall i \leq m)\ \types T_i \subtype T \infers C_i'\\[1mm]
          \Gamma \types e : \ul{d}\ \vv{\ul{T}} \infers \inj_X\,\ul{d}\ \vv{T},\ C_0 \\[1mm]
          (\forall i \leq m)\ \Gamma \cup \vv{x_{i}}:(\inj_X \vv{A})[\vv{T/\alpha}] \types e_i \infers T_i,\,C_i
        \end{array}}
      }
      {
        \begin{array}{l}
        \Gamma \types \matchtm{e}{\mid_{i=1}^m k_i\ \vv{x_i} \mapsto e_i} : \ul{T} \infers T,\,C
        \end{array}
      }\\[4mm]
    \end{array}
  \]
  \caption{Inference for expressions.}\label{fig:inf-exp}
  \end{figure}

\begin{definition}[Inference]
Inference is split into three parts: for subtyping (Figure \ref{fig:inf-sub}), for expressions (Figure \ref{fig:inf-exp}) and for modules (Figure \ref{fig:inf-mod}) using three judgement forms, respectively:
\[
    \types T_1 \subtype T_2 \infers C \qquad
    \Gamma \types e : \ul{T} \Longrightarrow T,\,C \qquad
    \Gamma \types m \Longrightarrow \Gamma',\,C
\]
Given two (extended) types $T_1$ and $T_2$ we infer a set of constraints $C$ under which the former will be a subtype of the latter using the system of judgements $\types T_1 \subtype T_2 \infers C$.
For expressions in context $\Gamma \types e : \ul{T}$, we infer (extended) monotypes $T$ and the constraints $C$ under which they are permissible using a system of judgements of the form $\Gamma \types e : \ul{T} \Longrightarrow T,\,C$.
In such judgements, $\Gamma$ a constrained type environment, i.e. a finite map from term variables to constrained type schemes.  We will omit the underlying type when not important.
The rules are given in Figure~\ref{fig:inf-exp}.
Constrained refinement type schemes are inferred for module-level definitions using a system of judgements of shape $\Gamma \types m \Longrightarrow \Gamma',\,C$.
The definitions are given in Figure~\ref{fig:inf-mod}.
The systems can be read algorithmically by regarding the quantities before the $\Longrightarrow$ as inputs the quantities afterwards as outputs (however, it should be noted that, assuming regular datatypes only, the subtyping relation must be computed by achieving a fixed point explicitly).
\end{definition}

Constrained type generation via these systems of rules follows a well established pattern for expressions and modules (see e.g. \cite{odersky-sulzmann-wehr-TSPOS1999} for a general treatment of the non-refinement case), so we concentrate on the inference rules for subtyping.
Like the more standard inference rules for expressions and modules, the inference rules for subtyping generate a derivation tree and a system of constraints whose solution guarantees the correctness of the corresponding instance of the derivation tree.
However, in the case of subtyping, the derivation tree is not a proof in the system of Figure \ref{fig:subtype-simple}, which is for the complement of the subtyping relation, but rather a proof that the solution constitutes a simulation in the sense of Lemma \ref{lem:simulation}.
For example, the conclusion of \rlnm{ISData} yields the constraints $\{X(d) \subseteq Y(d)\} \cup \bigcup_{k\in\mathbb{K}} \bigcup_{i=1}^{n} (k \in X(\ul{d}))\ ?\ C_{k_i}$.
The first part of this constraint encodes the first bullet of Lemma \ref{lem:simulation}: the environment $\canon{\Delta}$ at $\inj_Y\,d$ must include all constructors included by the same environment at $\inj_X\,d$.
Since, for any refinement $X$, $\dom(\canon{\Delta}(\inj_X\,d)) = X(d)$ (recall the notational abuse adopted at the start of Section \ref{sec:constraints}), we arrive at $X(d) \subseteq Y(d)$.
The second part of this constraint encodes the second bullet of the lemma: \emph{if} $k \in \dom(\canon{\Delta}(\inj_X\,d))$ (and, therefore, $k \in \dom(\canon{\Delta}(\inj_Y\,d))$), then the corresponding argument types are again related --- in inference we recursively infer constraints on the relationship between the types and guard the constraints by $k \in X(d)$.

\begin{theoremrep}[Soundness and completeness of $\subtype$-inference]\label{thm:correctness-subtype-inf}
  Let $T_1$ and $T_2$ be extended types and 
  suppose $\types T_1 \subtype T_2 \infers C$.  Then, for all assignments $\theta$: $\types \theta T_1 \subtype \theta T_2 \ \textit{iff}\  \theta \models C$.
\end{theoremrep}
\begin{proof}
  The proof is by induction on the inference judgement.
  \begin{description}
    \item[\rlnm{ISBase} \rlnm{ISTyVar}] Obvious.
    \item[\rlnm{ISArr}] 
      In the forward direction, suppose $\types \theta T_{11} \to \theta T_{12} \subtype \theta T_{21} \to \theta T_{22}$.  By inversion, necessarily $\types \theta T_{21} \subtype T_{11}$ and $\theta T_{12} \subtype T_{22}$.  In this case, $C = C_1 \cup C_2$.  By the induction hypothesis, $\theta \models C_1$ and $\theta \models C_2$, so $\theta \models C$ as required.\\[1mm]
      In the backward direction, suppose $\theta \models C_1 \cup C_2$.  Then $\theta \models C_1$ and $\theta \models C_2$ and it follows from the induction hypothesis that $\types \theta T_{21} \subtype T_{11}$ and $\theta T_{12} \subtype T_{22}$.  Hence, \rlnm{SArr} is satisfied and so $\types \theta T_{11} \to \theta T_{12} \subtype \theta T_{21} \to \theta T_{22}$.
    \item[\rlnm{ISData}]
      In the forward direction, suppose $\types \inj_{\theta(X)}\,d\ \vv{T_1}  \subtype \inj_{\theta(Y)}\,d\ \vv{T_2}$ and suppose $\ul{\Delta}(d)(k)$ is defined.
      Then it follows from \rlnm{SMis} that $\dom(\canon{\Delta}(\inj_{\theta(X)}\,d)) \not\subseteq \dom(\canon{\Delta}(\inj_{\theta(Y)}\,d))$.
      Note that, by definition, $\dom(\canon{\Delta}(\inj_{\theta(Z)}\,d)) = \theta(Z)(d)$, for any refinement variable $Z$.  Hence, $\theta \models X(d) \subseteq Y(d)$.
      To see that, also, $\theta \models k \in X(d)\ ?\ C$, assume $\theta \models k \in X(d)$.
      Therefore, also $\canon{\Delta}(\inj_{\theta(X)}\,d))(k)$ and $\canon{\Delta}(\inj_{\theta(Y)}\,d))(k)$ are defined.
      Let us fix $\canon{\Delta}(\inj_{\theta(X)}\,d))(k) = \forall \vv{\alpha}.\ A_{X_1} \to \cdots A_{X_m} \to d\ \vv{\alpha}$, and $\canon{\Delta}(\inj_{\theta(Y)}\,d))(k) = \forall \vv{\alpha}.\ A_{Y_1} \to \cdots A_{Y_m} \to d\ \vv{\alpha}$.
      Hence, by \rlnm{SSim}, $\types A_{X_i}[\vv{T_1/\alpha}] \subtype A_{Y_i}[\vv{T_2/\alpha}]$ for each $i \in [1 .. m]$.
      Note that, by definition, $\canon{\Delta}(\inj_{Z}\,d))(k) = \inj_{Z}(\ul{\Delta}(k))$ for any refinement variable $Z$.
      It follows from the induction hypothesis, therefore, that $\theta \models C$, as required.\\[2mm]
      In the backward direction, suppose (i) $\theta \models (k \in X(d)\ ?\ C)$ and (ii) $\theta \models X(d) \subseteq Y(d)$.
      Then, by (ii), \rlnm{SMis} is satisfied.
      To see that \rlnm{SSim} is also satisfied, suppose $\canon{\Delta}(\theta(X))(k)$ and $\canon{\Delta}(\theta(Y))(k)$ are defined.
      Then, by (i), $\theta \models C$ and it follows from the induction hypothesis that $\types A_{X_i}[\vv{T_1/\alpha}] \subtype A_{Y_i}[\vv{T_2/\alpha}]$.
      Since, by definition, $\canon{\Delta}(\inj_{Z}\,d))(k) = \inj_{Z}(\ul{\Delta}(k))$ for any refinement variable $Z$, \rlnm{SSim} is satisfied and so it must be that $\types \inj_{\theta(X)}\, d \subtype \inj_{\theta(Y)}\, d$ as required.
  \end{description}
\end{proof}

\begin{figure}
  \[
    \begin{array}{c}
    \prftree[l]{\rlnm{IModE}}{\Gamma \types \epsilon \Longrightarrow \Gamma}
    \\[6mm]
    \prftree[l,r]{\small$\left|\begin{array}{l}\fresh_{\ul{T}}(T)\\[1mm]\vv{X} = \frv(T)\end{array}\right.$}{\rlnm{IModD}}{
      \Gamma \types m \Longrightarrow \Gamma'
    }
    {
      \Gamma' \cup \{x : T\} \types e \Longrightarrow T',\,C_1
    }
    {
      \types T' \subtype T \infers C_2
    }
    {
      \Gamma \types m \cdot \langle x: \forall \vv{\alpha}.\ \ul{T} = \Abs{\vv{\alpha}}{e} \rangle \Longrightarrow \Gamma' \cup \{x : \forall \vv{\alpha}.\ \forall \vv{X}.\ C_1 \cup C_2 \implies T \}
    }
    \end{array}
  \]
\caption{Inference for modules.}\label{fig:inf-mod}
\end{figure}

\begin{toappendix}

  Additionally, we define the equivalence of two refinement variable assignments $\sigma$ and $\theta$ modulo a set of refinement variables $S$, written $\sigma \equiv_S \theta$, by requiring that they are identical on the variables in $S$, i.e. $\forall X \in S.\ \sigma(X) = \theta(X)$.

\begin{lemma}\label{lem:term-sound}
  Let $\Gamma$ be a closed constrained type environment, $\Gamma'$ a type environment, $e$ be an expression, $V \in \ETy$, and $C$ a constraint set.
  Suppose $\Gamma \types e \infers V,\,C$, $\Gamma' \subtype \bananas{\Gamma}$, and there exists a refinement variable substitution $\theta$ such that $\theta \models C$, then there exists a type $T \in \Ty\ \canon{D}$ such that:
  \begin{enumerate}[(i)]
    \item $\Gamma' \types e : T$
    \item $\types T \subtype \theta V$
  \end{enumerate}
\end{lemma}
\begin{proof}
    The proof is by induction of $\infers$.
    \begin{description}
      \item[\rlnm{ICst}]
        In this case, $e$ is of shape $c$ and we may assume $\mathbb{C}(c) = V$.
        Let $T = \mathbb{C}(c)$.
        By \rlnm{TCst}, we have that $\Gamma' \types e : T$.
        As $V$ is an unrefined type, it is invariant under any substitution, i.e. $\theta V = V$, and hence (ii) is trivial satisfied.

      \item[\rlnm{ICon}]
        In this case, $e$ is of shape $k$, and $V = \inj_X(\ul{T_1}) \rightarrow \cdots{} \rightarrow \inj_X(\ul{T_m}) \rightarrow \inj_X(d)$ for fresh $X$ where $\ul{\Delta}(d)(k) = \ul{T_1} \rightarrow \cdots{} \rightarrow \ul{T_m}$.
        The constraint set $C$ is $\{ k \in X(d) \}$.
        We know, therefore, that $\theta$ is a substitution mapping $X$ to a constructor choice function $f$, such that $k \in \dom(\Delta^*(\inj_f(d))$.
        Let $T_i = \Delta^*(\inj_f(d))(k)(i)$ and $T = T_1 \rightarrow \cdots{} \rightarrow T_n \rightarrow \inj_f(d)$.
        By the definition of $\Delta^*$, $\theta V = T$.
        Finally, by \rlnm{TCon}, we have that $\Gamma' \types k : T$ as required.

      \item[\rlnm{IVar}]
        In this case, $e$ is of shape $x\ \vv{\ul{T}}$.
        Let us assume that $x:\forall \vv{\alpha}.\ \forall \vv{X}.\ C \implies U \in \Gamma$.
        Hence, for fresh $\vv{Y}$ and $\vv{T}$, $V$ is $U[\vv{Y/X}][\vv{T/\alpha}]$, and $\theta$ solves $C[\vv{Y/X}]$.
        We may assume $x:\forall \vv{\alpha}.\ \sigma U \in \Gamma'$, for some $\sigma$.
        Thus by \rlnm{IVar} $\Gamma' \types x \vv{\ul{T}} : (\sigma U)[\vv{T/\alpha}]$.
        Define a new model $\theta'$ as follows:
        \[
          \theta'(Z) =
            \begin{cases}
              \theta(Y_i) &\textrm{if}\ Z = X_i\\
              \sigma(X_i) &\textrm{if}\ Z = Y_i\\
              T_i         &\textrm{if}\ Z = T_i\\
              \sigma(Z)   &\textrm{otherwise}\
            \end{cases}
        \]
        The consistency of this definitions follows from the freshness of $\vv{T}$ and $\vv{V}$.
        Clearly this solves $C$, and $\theta' V = \theta' U[\vv{Y/X}][\vv{T/\alpha}] = (\sigma U)[\vv{T/\alpha}]$ trivially satisfying (ii).

      \item[\rlnm{IAbs}]
        In this case $e = \lambda x. e_1$, and $V = V_1 \to V_2$ where $V_1$ is fresh.
        Suppose $\theta \models C$.
        As $C$ is also the constraint set of the hypotheses, by induction, there exists a type $T_2$, such that $\Gamma' \cup \{x : T_1\} \types e_1 : T_2$, $T_2 \subtype \theta V_2 $, and $\theta V_1 = T_1$.
        By \rlnm{TAbs} we may derive $\Gamma' \types \lambda x.\ e_1 : T_1 \rightarrow T_2$.
        Let $T = T_1 \to T_2$, then clearly we have $\types T \subtype \theta V$ as required.

      \item[\rlnm{IApp}]
        In this case $e = e_1\ e_2$ and $C = C_1 \cup C_2 \cup C_3$.
        Hence by the rules of inference, necessarily $\Gamma \types e_1 \infers V_1 \to V_2,\ C_1$, $\Gamma \types e_2 \infers V_3,\ C_2$, and $\types V_3 \subtype V_1 \infers C_3$.
        Suppose $\theta \models C$.
        As this is also a solution to $C_3$ by Theorem \ref{thm:correctness-subtype-inf}, we have that $\theta V_3 \subtype \theta V_1$.
        By the induction hypothesis $\Gamma' \types e_1 : T_1 \to T_2$ and $\Gamma' \types e_2 : T_3$ with $\theta V_1 \subtype T_1$, $T_2 \subtype \theta V_2$, and $T_3 \subtype \theta V_3$.
        Using the transitivity of subtyping, we have that $T_3 \subtype T_1$.
        Hence by \rlnm{TSub} $\Gamma' \types e_2 : T_1$, and so by \rlnm{TApp} $\Gamma' \types e_1\ e_2 : T_2$ as required.

      \item[\rlnm{ICase}]
        In this case $e = \matchtm{e_0}{\mid_{i=1}^m k_i\ \vv{x_i} \mapsto e_i}$, and $C = C_0 \cup \bigcup_{i=1}^m k_i \in X(d)\ ?\ (C_i \cup C_i') \cup \{X(d) \subseteq \{k_1,\ \cdots{},\ k_m\}\}$.
        Let us suppose that $\theta$ maps $X$ to the constructor choice function $f$.
        It follows from inference that $\Gamma \types e \infers \inj_X(d),\ C_0$, and as $\theta \models C_0$, thus $\Gamma' \types e : T_0$ such that $\types T_0 \subtype \theta \inj_f(d)$.

        For each $i \leq m$, let us consider the case when $k_i \in \dom(\Delta^*(\inj_f(d)))$.
        By induction $\Gamma' \cup \{ \vv{x_i} : \Delta^*(\inj_f(d))(k_i) \} \types e_i : T_i$ for some $T_i \subtype \theta V_i$.
        As $C_i'$ must be satisfied by $\theta$, we also have that $\theta V_i \subtype \theta V$, and so $\Gamma' \cup \{ \vv{x_i} : \Delta^*(\inj_f(d))(k_i) \} \types e_i : \theta V$.

        As $\dom(\Delta^*(\inj_f(d))) \subseteq \{k_1, \cdots{}, k_m\}$, by \rlnm{TCase} we have that $\Gamma' \types e : \theta V$ as required.
    \end{description}
\end{proof}

\begin{lemma}\label{lem:term-complete}
  Let $\Gamma$ be a constrained type environment, $\Gamma'$ a type environment, $e$ be an expression, $T \in \Ty\ \canon{D}$, $V \in \ETy$, $\sigma$ a refinement variable substitution with domain $\frv(\Gamma)$ and $C$ a constraint set.
  Suppose $\Gamma' \types e : T$ and $\Gamma \types e \infers V,\,C$ and $\bananas{\sigma\Gamma} \subtype \Gamma'$.  Then there is $\theta$ such that:
  \begin{enumerate}[(i)]  
    \item $\sigma \equiv_{\frv(\Gamma)} \theta$
    \item and $\types \theta V \subtype T$
    \item and $\theta \models C$.
  \end{enumerate}
\end{lemma}
\begin{proof}
  The proof is by induction on $\types'$.
  \begin{description}
    \item[\rlnm{TVar}]
      In this case $e$ is of shape $x\ \vv{\ul{T}}$.  
      Assume $\under(\vv{T}) = \vv{\ul{T}}$, $(x:\forall \vv{\alpha}.\ U) \in \Gamma'$ and $U[\vv{T/\alpha}] = T$.
      Assume $\bananas{\sigma\Gamma} \subtype \Gamma'$.
      By definition, in $\Gamma$ there is some $x:\forall \vv{\alpha}.\ \forall \vv{X}.\ C' \implies U'$ and there is some $\tau$ such that $\types \tau(\sigma U') \subtype U$ and $\tau \models \sigma C'$.
      By definition, $V$ is of shape $U'[\vv{Y/X}][\vv{T'/\alpha}]$ for fresh $\vv{Y}$ and fresh types $\vv{T'}$.  Moreover, $C$ is of the form $C'[\vv{Y/X}]$.  
      Since the $\vv{T'}$ are fresh, there is a substitution $\theta'$ such that $\theta'\vv{T'} = \vv{T}$.
      Define $\theta$ as follows:
      \[
        \theta(Z) \coloneqq 
          \begin{cases}
            \theta'(Z) & \text{if $Z \in \frv(\vv{T'})$} \\
            \sigma(Z) & \text{if $Z \in \frv(\Gamma)$} \\
            \tau(X_i) & \text{if $Z = Y_i \in \vv{Y}$} \\
            \tau(Z) & \text{otherwise}
          \end{cases}
      \]
      Note that the freshness of $\vv{Y}$ and $\vv{T'}$ ensure the exclusivity of the four cases and requirement (i) of the theorem is satisfied.
      Observe that:
      \[
        \theta(U'[\vv{Y/X}][\vv{T'/\alpha}]) = (\theta (U'[\vv{Y/X}]))[\theta\vv{T'}/\alpha] = (\theta (U'[\vv{Y/X}]))[\vv{T}/\alpha]
      \]
      Next, since $\theta(X_i) = \theta(Y_i)$ for any $X_i \in \vv{X}$ and the codomain of $\tau$ is closed, it follows that:
      \[
        \theta (U'[\vv{Y/X}]) = (\theta U')[\theta(\vv{Y})/\vv{X}] = (\theta U')[\theta(\vv{X})/\vv{X}] = \theta U'
      \]
      Finally, by the disjointness of the cases, the fact that the codomain of $\sigma$ is necessarily closed, and $\frv(U') \subseteq \vv{X} \cup \frv(\Gamma)$, $\theta U' = \tau(\sigma U')$.
      Hence, overall $\theta(U'[\vv{Y/X}][\vv{T'/\alpha}]) = \tau(\sigma U')[\vv{T}/\alpha]$ and it follows by definition of subtyping that $\types \tau(\sigma U')[\vv{T}/\alpha] \subtype U[\vv{T}/\alpha]$ so that requirement (ii) is satisfied.
      Finally, note that $\theta \models C'[\vv{Y}/\vv{X}]$ iff $\theta \models C'$ since $\theta(Y_i) = \theta(X_i)$.
      Then, by definition, and the closedness of the codomain of $\sigma$, $\theta \models C'$ iff $\theta \models \sigma C'$.
      Then, since $\frv(C') \subseteq \vv{X} \cup \frv(\Gamma)$, it follows that $\theta \models \sigma C'$ iff $\tau \models \sigma C'$, which was an assumed.
      Hence, $\theta \models C'[\vv{Y}/\vv{X}]$ and requirement (iii) is satisfied.   

    \item[\rlnm{TCst}]
      In this case, $e$ is of shape $c$ and we may assume $\mathbb{C}(c) = T$.  Since $V=\mathbb{C}(c)$ and $C=\emptyset$, Any assignment $\theta$ extending $\sigma$ will satisfy the three requirements.

    \item[\rlnm{TCon}] 
      In this case, $e$ is a constructor $k\ \vv{\ul{T}}$ and $T$ is of the form $T_1 \to \cdots{} T_m \to \inj_f\,d\ \vv{T}$ with $T_1 \to \cdots \to T_m \to \inj_d\,d\ \vv{T} = \canon{\Delta}(\inj_g\,d)(k)[\vv{T/\alpha}]$ and $k \in \dom\,(\canon{\Delta}(\inj_f\,d))$ for some fresh $\vv{T}$.
      Hence, by definition, $T = \inj_f\,(\ul{\Delta}(d)(k))$ (*) and $k \in f(d)$ (***).
      In this case, $V$ is $\inj_X\,\ul{T}$ with $\ul{T} = \ul{\Delta}(d)(k) \to d$ and $X$ fresh.
      Moreover, $C$ is $\{k \in X(d)\}$  
      Therefore, we can define $\theta$ as follows:
      \[
        \theta(Z) = 
          \begin{cases}
            f & \text{if $Z = X$} \\
            \sigma(Z) & \text{otherwise} \\
          \end{cases}
      \]
      By freshness of $X$, this guarantees requirement (i).  
      By (*) above, we have $\inj_f\,\ul{T} \subtype T$, ensuring requirement (ii).
      Finally, the requirement (iii) is satisfied by (***).

    \item[\rlnm{TAbs}]
      In this case, $e$ is an abstraction $\abs{x:\ul{T_1}}{e'}$ and $T$ is of shape $T_1 \to T_2$.  We may assume that $\under(T_1) = \ul{T_1}$ and $x \notin \dom\,\Gamma$.
      Assume $\types \bananas{\sigma\Gamma} \subtype \Gamma'$.
      From inference we have $V$ necessarily of shape $V_1 \to V_2$ with $\fresh_{\ul{T_1}}(V_1)$ and $\Gamma \cup \{x:V_1\} \types e \infers T_2,\,C$.  
      By the freshness of $V_1$, it follows that there is some $\sigma'$ such that $\sigma'V_1 = T_1$ and $\sigma'(Z) = \sigma(Z)$ on any $Z \notin \frv(V_1)$.  
      Hence, $\bananas{\sigma\Gamma} \cup \{x:T_1\} = \bananas{\sigma'(\gamma \cup \{x:V_1\})}$ and, by definition, $\types \bananas{\sigma\Gamma} \cup \{x:T_1\} \subtype \Gamma' \cup \{x:T_1\}$.
      Therefore, it follows from the induction hypothesis that there is an assignment $\theta$ satisfying (a) $\theta \equiv_{\frv(\Gamma \cup \{x:V_1\})} \sigma'$, (b) $\types \theta V_2 \subtype T_2$ and (c) $\theta \models C$.  Then this $\theta$ works also as a witness to the main result since (a) implies $\theta \equiv_{\frv(\Gamma)} \sigma$, (a) and (b) together imply $\types \theta V_1 \to V_2 \subtype T_1 \to T_2$ and (c) is exactly requirement (iii).

    \item[\rlnm{TApp}]
      In this case, $e$ is an application $e_1\,e_2$.  From inference we have, necessarily, $\Gamma \types e_1 \infers V_1 \to V_2,\,C_1$, $\Gamma \types e_2 : V_3,\,C_2$ and $\types V_3 \subtype V_1 \infers C_3$.  
      It follows immediately from the induction hypothesis that there are assignments $\theta_1$ and $\theta_2$ such that (a) $\theta_i \equiv_{\frv(\Gamma)} \sigma$, (b1) $\types \theta_1 V_1 \to V_2 \subtype T_1 \to T_2$, (b2) $\types \theta_2 V_3 \subtype T_1$, (c1) $\theta_1 \models C_1$ and (c2) $\theta_2 \models C_2$.  
      Define $\theta$ as follows:
      \[
        \theta(Z) =
          \begin{cases}
            \theta_1(Z) & \text{if $Z \in \frv(V_1 \to V_2)$} \\
            \theta_2(Z) & \text{if $Z \in \frv(V_3)$} \\
            \sigma(Z) & \text{otherwise}
          \end{cases}
      \]
      This is well defined since one can easily verify by inspection that the inference system guarantees that the only overlap in refinement variables between sibling branches is in the free refinement variables of the environment, and $\theta_1$ and $\theta_2$ agree on this by (a).  
      This construction therefore satisfies requirement (i).  
      Furthermore, (b1) implies $\types \theta V_2 \subtype T_2$.
      Finally, it follows from (c1) and (c2) that $\theta \models C_1 \cup C_2$.
      By (b1) and (b2) we have $\types \theta V_3 \subtype T_1$ and $\types T_1 \subtype \theta V_1$ so, by transitivity, also $\types \theta V_3 \subtype \theta V_1$.
      Hence, it follows from the completeness of subtype inference (Theorem~\ref{thm:correctness-subtype-inf}) that, therefore, $\theta \models C_3$.

    \item[\rlnm{TCase}]
      In this case, $e$ is of shape $\matchtm{e'}{\shortmid_{i=1}^m k_i\ \vv{x_i} \mapsto e_i}$ and we may assume that $\{i_1,\ldots,i_n\} \subseteq \{1,\ldots,m\}$ and $\dom(\canon{\Delta}(d)) = \{k_{i_1},\ldots,k_{i_n}\}$ (*).
      We may also assume that $d = \inj_f\,{\ul{d}}$ for some choice $f$.
      From inference, we have necessarily $\Gamma \types e : \infers \inj_X\,\ul{d},\,C_0$ and, for all $i \in \{1\ldots m\}$, $\types V_i \subtype V \infers C_i'$ (**) and $\Gamma \cup \{\vv{x_i} : \inj_X(\ul{\Delta}(\ul{d})(k_i))\} \types e_i \infers V_i,\,C_i$; where $X$, $V$ and each $V_i$ are fresh.  
      Since $V$ is fresh, there is some substitution $\theta'$ such that $\theta'\,V = T$ (***).
      It follows from the induction hypothesis that there is $\theta_0$ such that (a0) $\theta_0 \equiv_{\frv(\Gamma)} \sigma$, (b0) $\types \theta_0\,\inj_X\,\ul{d} \subtype d$ and (c0) $\theta_0 \models C_0$.
      Set $\sigma'(X) = \theta_0(X)$ and $\sigma'(Z) = \sigma(Z)$ for all other $Z$.
      Then, for all $j \in \{1\ldots n\}$: 
      \[
        \bananas{\sigma \Gamma} \cup \{\vv{x_{i_j}} : \canon{\Delta}(\inj_{\theta_0(X)}\,\ul{d})(k_{i_j})\} = \bananas{\sigma' (\Gamma \cup \{\vv{x_{i_j}} : \inj_X\,(\ul{\Delta}(\ul{d})(k_{i_j}))\}}
      \]
      By (b0) and \rlnm{SSim}, it follows that $\types \bananas{\sigma \Gamma} \cup \{\vv{x_{i_j}} : \canon{\Delta}(\inj_{\theta_0(X)}\,\ul{d})(k_{i_j})\} \subtype \bananas{\sigma \Gamma} \cup \{\vv{x_{i_j}} : \canon{\Delta}(d)(k_{i_j})\}$.
      Hence it follows from the induction hypothesis that there are substitutions $\theta_{i_j}$ (for each $j \in \{1\ldots n\}$) such that (ai) $\theta_{i_j} \equiv_{\frv(\Gamma \cup \{\vv{x_{i_j}} : \inj_X\,(\ul{\Delta}(\ul{d})(k_{i_j}))\})} \sigma'$, (bi) $\theta_{i_j}\,V_i \subtype T$ and (ci) $\theta_{i_j} \models C_{i_j}$.
      Define $\theta$ as follows:
      \[
        \theta(Z) = 
          \begin{cases}
            \theta'(Z) & \text{if $Z \in \frv(V)$} \\
            \theta_0(Z) & \text{if $X = Z$ or $Z \in \frv(C_0)$} \\
            \theta_{i_j}(Z) & \text{if $Z \in \frv(V_{i_j}) \cup \frv(C_{i_j})$} \\
            \sigma(Z) & \text{otherwise}
          \end{cases}
      \]
      The use of freshness guarantees the well-definedness since refinement variables introduced in different branches are distinct and so the only variables shared are those in $\Gamma$, on which all agree by (a0) and (ai).
      Hence, requirement (i) is satisfied.
      It follows from (***) that $\theta\,V = T$ and by (bi) $\types \theta\,V_{i_j} \subtype T$.
      Hence, $\types \theta\,V_{i_j} \subtype \theta\,V$ and it follows from the completeness of subtype inference, Theorem~\ref{thm:correctness-subtype-inf}, that $\theta \models C'_{i_j}$.
      By (*) and \rlnm{SMis}, $\dom\,(\canon{\Delta}(\inj_{\theta(X)}\,\ul{d})) \subseteq \{k_{i_1},\ldots,k_{i,n}\}$ and so $\theta(X)(\ul{d}) \subseteq \{k_{i_1},\ldots,k_{i,n}\}$.
      Consequently, by the foregoing and (ci):
      \[
        \theta \models \{X(d) \subseteq \{k_i,\ldots,k_m\}\} \bigcup_{i=1}^m k_i \in X(\ul{d}) ? (C_i \cup C_i')
      \]
      Finally, requirement (iii) is satisfied by also taking into account (c0).

    \item[\rlnm{TSub}] 
      In this case we may assume $\types T_1 \subtype T_2 = T$.  It follows from the induction hypothesis that there is $\theta$ such that (i) $\theta \equiv_{\frv(\Gamma)} \sigma$, (ii) $\types \theta\,V \subtype T_1$ and (iii) $\theta \models C$.  Since, by transitivity, $\types \theta\,V \subtype T_2$, the same $\theta$ acts as witness to the result.
 \end{description}
\end{proof}

\end{toappendix}

The following states the correctness of type inference for expressions in a closed environment (e.g. for module-level definitions).  The appendix contains a proof for the general case.

\begin{theoremrep}[Soundness and completeness of expression inference]\label{thm:inf-exp-correctness}
  Let $\Gamma$ be a constrained type environment, $e$ an expression, $V$ an extended refinement, $C$ a set of constraints and let $\Gamma \types e \infers V,\,C$.  Then, for all refinement types $T$: $\bananas{\Gamma} \types e : T  \ \textit{iff}\  \exists \theta.\; {\theta \models C} \;\wedge\; {\types \theta V \subtype T}$.
\end{theoremrep}
\begin{proof}
  Follows immediately from Lemmas \ref{lem:term-sound} and \ref{lem:term-complete}.
\end{proof}

\begin{toappendix}

\begin{lemma}\label{lem:mod-inf-soundness}
  Let $\Gamma$, and $\Gamma'$ be closed constrained type environments, $\Delta \subtype \bananas{\Gamma}$, and $m$ a module.
  If $\Gamma \types m \infers \Gamma'$, and $\bananas{\Gamma'}$ isn't empty, then there exists a $\Delta' \subtype \bananas{\Gamma'}$ such that $\Delta \types m : \Delta'$.
\end{lemma}
\begin{proof}
    The proof is by induction of $\infers$.
  \begin{description}
    \item[\rlnm{IModE}]
      In this case the module is empty, i.e. $m=\epsilon$, and $\Gamma' = \Gamma$.
      Let $\Delta' = \Delta$ which is a subenvironment of $\bananas{\Gamma'}$.
      By \rlnm{TModE} we can derive $\Delta \types m : \Delta'$ as required.

    \item[\rlnm{IModD}]
      If the module has shape $m \cdot \langle x: \forall \vv{\alpha}.\ \ul{T} = \Abs{\vv{\alpha}}{e} \rangle$, and $\Gamma' = \Gamma'' \cup \{x : \forall \vv{\alpha}.\ \forall \vv{X}.\ \restr{\Sat(C_1 \cup C_2)}{\vv{X}} \implies V \}$.
      Let us assume $\Gamma'' \cup \{x : V\} \types e \infers V',\ C_1$, and $\types V' \subtype V \infers C_2$ for some fresh $V$ with $\vv{X} = \frv(V)$.
      The induction hypothesis provides a $\Delta'' \subtype \bananas{\Gamma''}$ such that $\Delta \types m : \Delta''$.
      As $\bananas{\Gamma'}$ isn't empty, there must exist a solution to $\restr{\Sat(C_1 \cup C_2)}{\vv{X}}$.
      By lemma \ref{thm:equi-consistency} we therefore have that both $C_1$ and $C_2$ are solvable by some refinement variable substitution $\theta$.
      Additionally, there must some type $T' \subtype \theta V'$ such that $\Delta' \cup \{ x : T \} \types e : T'$ where $T = \theta V$ by lemma \ref{lem:term-sound}.
      As $\theta V' \subtype \theta V$ we have that $T' \subtype T$.
      Thus by \rlnm{TModE} we have that $\Delta \types m \cdot \langle x: \forall \vv{\alpha}.\ \ul{T} = \Abs{\vv{\alpha}}{e} \rangle : \Delta' \cup \{x : \forall \vv{\alpha}.\ T \}$ as required.
  \end{description}
\end{proof}

\begin{lemma}\label{lem:mod-inf-completeness}
  Let $\Gamma$ and $\Gamma'$ be constrained type environments, $\Delta$ and $\Delta'$ type environments, and $m$ a module.
  Suppose $\bananas{\Gamma} \subtype \Delta$, and $\Delta \types m : \Delta'$ and $\Gamma \types m \infers \Gamma'$, then $\bananas{\Gamma'} \subtype \Delta'$.
\end{lemma}
\begin{proof}
    The proof is by induction of $\types$.
  \begin{description}
    \item[\rlnm{TModE}]
      If the module is empty, then $\Delta = \Delta'$ and $\Gamma = \Gamma'$.
      Therefore $\Delta' \subtype \bananas{\Gamma'}$ trivially follows from the hypothesis.

    \item[\rlnm{TModD}]
      Suppose the module is of shape $m \cdot \langle x = \Abs{\vv{\alpha}}{e} \rangle$, and therefore $\Delta' = \Delta'' \cup \{x : \forall \vv{\alpha}.\ T\}$ for some $\Delta''$ and $T$ such that $\Delta \types m : \Delta''$, $\Delta'' \cup \{ x : \forall \vv{\alpha}.\ T\} \types e : T'$, and $\types T' \subtype T$.
      As $\Gamma \types m \cdot \langle x = \Abs{\vv{\alpha}}{e} \rangle \infers \Gamma'$, we know that $\Gamma'$ has shape $\Gamma'' \cup \{ x:\forall \vv{\alpha}.\ \forall \vv{X}.\ C_1 \cup C_2 \implies V\}$ such that $\Gamma \types m \infers \Gamma''$, $\Gamma'' \cup \{x : V\} \types e \infers V', C_1$, and $\types V' \subtype V \infers C_2$.
      By the induction hypothesis, therefore, we have that $\bananas{\Gamma''} \subtype \Delta''$.
      As $V$ is fresh we can instantiate lemma \ref{lem:term-complete} with a substitution that maps it to $T'$.
      Therefore, there is some $\theta$ that solves $C_1$ with $\theta V' \subtype T'$, and $\theta V = T$.
      Hence $\Delta'$ is $\Gamma'$ with $x$ instantiated by $\theta$ all other variables instantiated as in $\Delta''$.
      As $\theta V' \subtype T'$ trivially holds, we also have that $\theta \models C_2$.
      Thus $\bananas{\Gamma'} \subtype \Delta'$ as required.

  \end{description}
\end{proof}
\end{toappendix}

Finally, we can state the overall correctness of inference for modules.

\begin{theoremrep}[Soundness and completeness of module inference]\label{thm:inf-mod-correctness}
  Suppose $\Gamma$ and $\Gamma'$ are closed constrained type environments, $m$ a module and $\Gamma \types m \infers \Gamma'$.  Then, for all type environments $\Delta$:
  \[
    \bananas{\Gamma} \types m : \Delta \quad{\textit{iff}}\quad \bananas{\Gamma'} \subtype \Delta
  \]
\end{theoremrep}
\begin{proof}
  Follows immediately from Lemmas \ref{lem:mod-inf-soundness} and \ref{lem:mod-inf-completeness}.
\end{proof}

\section{Saturation}\label{sec:saturation}

The solvability of constraints can be determined by a process of saturation under all possible consequences.  This is a generalisation of the transitive closure of simple inclusion constraint graphs, and a particular instance of Horn clause resolution more generally.  For our constraint language, saturated constraint sets have a remarkable property: they can be restricted to any subset of their variables whilst preserving solutions.



\begin{definition}[Atomic constraints]
 A constraint is said to be \emph{atomic} just if its body is one of the following four shapes:
  \[
    X(\underline{d}) \subseteq Y(\underline{d})
    \qquad
    X(\underline{d}) \subseteq \{k_1,\ldots,k_m\}
    \qquad
    k \in X(\underline{d})
    \qquad
    k \in \emptyset
  \]
  An atomic constraint is said to be \emph{trivially unsatisfiable} if it is of shape $\emptyset\ ?\ k \in \emptyset$.  A constraint set is said to be \emph{trivially unsatisfiable} just if it contains a trivially unsatisfiable constraint.
\end{definition}

  By applying standard identities of basic set theory, every constraint is equivalent to a set of atomic constraints.
  In particular, a constraint of the form $k \in \{k_1,\ldots,k_m\}$ is equivalent to the empty set of atomic constraints (i.e. can be eliminated) whenever $k$ is one of the $k_i$.

\begin{definition}[Saturated constraint sets]
  An atomic constraint set, i.e. one that only contains atomic constraints, is said to be \emph{saturated} just if it is closed under the saturation rules in Figure~\ref{fig:resolution}.
  We write $\Sat(C)$ for the saturated atomic constraint set obtained by iteratively applying the saturation rules to $C$.  Note, c.f. Remark \ref{rmk:Horn}, all three rules correspond to special cases of resolution.
\end{definition}

\begin{figure*}
  \[
    \begin{array}{c}
    \prftree[l]{\rlnm{Transitivity}}{
      \phi\ ?\ S_1 \subseteq S_2
    }
    {
      \psi\ ?\ S_2 \subseteq S_3
    }
    {
      \phi \cup \psi\ ?\ S_1 \subseteq S_3 
    }
    \\[6mm]
    \prftree[l]{\rlnm{Satisfaction}}{
      \phi\ ?\ k \in X(d)
    }
    {
      \psi,\, k \in X(d)\ ?\ S_1 \subseteq S_2
    }
    {
      \phi \cup \psi ?\ S_1 \subseteq S_2
    }
    \\[6mm]
    \prftree[l]{\rlnm{Weakening}}{
      \phi\ ?\ X(d) \subseteq Y(d)
    }
    {
      \psi,\, k \in Y(d)\ ?\ S_1 \subseteq S_2
    }
    {
      \phi \cup \psi,\, k \in X(d)\ ?\ S_1 \subseteq S_2
    }
    \end{array}
  \]
\caption{Saturation rules.}\label{fig:resolution}
\end{figure*}

The \rlnm{Transitivity} rule closes subset inequalities under transitivity, but must keep track of the associated guards by taking the union.  
The \rlnm{Satisfaction} rule allows for a guard atom $k \in X(d)$ to be dropped whenever the same atom constitutes the body of another constraint in the set (but the other guards from both must be preserved).
Finally, the \rlnm{Weakening} rule allows for replacing $Y(d)$ in a guard by $X(d)$ when it is known to be no larger, thus weakening the constraint.
Saturation under these rules preserves and reflects solutions:

\begin{theoremrep}[Saturation equivalence]\label{thm:equi-consistency}
  For any assignment $\theta$, $\theta \models C$ iff $\theta \models \Sat(C)$.
\end{theoremrep}
\begin{proof}
We shall show that the resolution rules preserve solutions (naturally they reflect solutions too, since they do not remove constraints).
In each case we shall assume the derived guard holds, otherwise there is nothing to show.
  \begin{itemize}
    \item Suppose $\phi\ ?\ S_1 \subseteq S_2$ and $\psi\ ?\ S_2 \subseteq S_3$ appear in $C$, and both $\theta$ satisfies both $\phi$ and $\psi$.
      Then we must have $\theta S_1 \subseteq \theta S_2$ and $\theta S_2 \subseteq \theta S_3$, as $\theta$ solves $C$.
      From transitivity of subset relation, it follows that $\theta S_1 \subseteq \theta S_2$. If the guards do not hold there is nothing to show.

    \item
      Suppose $\phi\ ?\ \dom(X(d)) \subseteq \dom(Y(d))$ and $\psi \cup (Y, \ul{d}) \mapsto k\ ?\ S_1 \subseteq S_2$ appear in $C$, and that $\theta$ satisfies $\phi \cup \psi \cup (X, \ul{d}) \mapsto k$.
      Therefore $\dom((\theta X)(d)) \subseteq \dom((\theta Y)(d))$ as $\theta$ solves $C$.
      Additionally, $k \in \dom((\theta X)(\ul{d}))$, and so $k \in \dom((\theta Y)(\ul{d}))$.
      Thus $\psi \cup (Y, \ul{d}) \mapsto k$ holds under $\theta$, and $\theta S_1 \subseteq \theta S_2$.

    \item
      Suppose $\phi\ ?\ k \in \dom(X(d))$ and $(X, \ul{d}) \mapsto k \cup \psi\ ?\ S_1 \subseteq S_2$ appear in $C$, and $\theta$ satisfies both $\phi$ and $\psi$.
      Then we have that $k \in \dom((\theta X)(d))$.
      Clearly the guard of $S_1 \subseteq S_2$ then must also hold and $\theta S_1 \subseteq \theta S_2$ as required.
  \end{itemize}
\end{proof}

\begin{toappendix}
\begin{definition}\label{def:theta}
  We say that $\tau$ is a \emph{partial solution} of some constraint set $C$, if $\tau$ solves the restriction of $C$ to the domain of $\tau$, i.e. $\tau$ solves $\restr{C}{\dom(\tau)}$.

  Additionally, we construct the \emph{extended solution} $\theta_\tau$ (or just $\theta$ when $\tau$ is implied) as follows.
  For each refinement variable $X$ not in the domain of $\tau$, define $(\theta X)(d)(k) = \ul{\Delta}(d)(k)$ whenever there exists some $\phi\ ?\ S \subseteq \dom(X(d)) \in C$, such that $\phi$ holds under $\tau$, and $k \in \tau S$.
\end{definition}

\begin{lemma}\label{lem:theta-tau}
  Suppose $\tau$ is a partial solution of a saturated constraint set $C$.
  For any constraint $\phi\ ?\ S_1 \subseteq S_2 \in C$ such that $\phi$ satisfied by $\theta_\tau \circ \tau$, then there is a constraint with the same body, i.e. of the form $\psi\ ?\ S_1 \subseteq S_2$, such that $\tau$ satisfies $\psi$.
\end{lemma}
\begin{proof}
  Our proof shall be by induction of the cardinality of $\{ (X,\ k) \mid X \not\in \dom(\tau) \}$.

  If the cardinality is zero, then $\tau$ must satisfy $\phi$ and so we are done.

  We now consider the case when the cardinality is $n + 1$, and let us assume the induction hypothesis holds for $n$.
  Suppose $(X,\ k) \in \phi = k$ and $X \not\in \dom(\tau)$.
  As $\theta_\tau \circ \tau$ satisfies $\phi$ we have that that $k \in \dom((\theta_\tau X)(\ul{d}))$ for some $\ul{d}$.
  And so by construction $\psi\ ?\ S \subseteq \dom(X(\ul{d})) \in C$ with $k \in \tau S$ and $\tau$ satisfying $\psi$.
  As $C$ is saturated $S$ must either be the singleton $k$, or of the form $\dom(Y(\ul{d}))$:
  \begin{itemize}
    \item
      In the first case we may apply the rule \rlnm{Satisfaction} as $C$ is saturated to conclude $\psi \cup \phi - (X,\ k)\ ?\ S_1 \subseteq S_2 \in C$.

    \item
      In the second case we may apply the rule \rlnm{Weakening} to conclude $\psi \cup \phi \cup (Y,\ k) - (X,\ k)\ ?\ S_1 \subseteq S_2 \in C$.
  \end{itemize}

  In either case we are now left with a guard which $\theta_\tau \circ \tau$ satisfies and has $n$ pairs with refinement variables that are not in the domain of $\tau$, hence we can apply the induction hypothesis.
\end{proof}

This lemma expresses the fact that the extended solution $\theta_\tau$ does not make any arbitrary choices --- whenever it satisfies the guard of a constraint, the constraint must already have been satisfied according to resolution.

\begin{lemma}\label{lem:sat-consistency}
  Suppose $C$ is a saturated constraint set which doesn't contain $\emptyset\ ?\ k \in \emptyset$, and suppose $\tau$ is a partial solution $C$.
  Then $\theta_\tau \circ \tau$ solves $C$.
\end{lemma}
\begin{proof}
  We need only consider the constraints in $C$ which reference some refinement variable not in the domain of $\tau$.

  Let $\phi\ ?\ S_1 \subseteq S_2$ be such a constraint.
  If $\theta_\tau \circ \tau$ doesn't hold on $\phi$ then there is nothing to show.
  On the other hand, we may apply the above lemme to deduce that there is a constraint $\psi\ ?\ S_1 \subseteq S_2 \in C$ with $\psi$ satisfied by $\tau$.

  We now consider the possible forms of $S_1$ and $S_2$:

  \begin{itemize}
    \item
      Suppose both $S_1$ and $S_2$ are in the domain of $\tau$.
      Then $\psi\ ?\ S_1 \subseteq S_2$ is in the restriction of $C$ to the domain of $\tau$ and as $\tau$ satisfies $\psi$ we have that $\tau S_1 \subseteq \tau S_2$ as required.

    \item
      Suppose $S_1$ is in the domain of $\tau$, but $S_2 = \dom(X(d))$ for some $X \not\in \dom(\tau)$.
      If $k \in \tau S_1$, then by construction $k \in \dom((\theta X)(d))$, and so the subset relation holds.

    \item
      Suppose $S_1 = \dom(X(d))$ for some $X \not\in \dom(\tau)$.
      Let $k \in \dom((\theta X)(d))$, we shall show that $k \in (\theta \circ \tau) S_2$.
      By construction $\psi'\ ?\ S \subseteq \dom(X(d)) \in C$ such that $k \in \tau S$ and for some $\psi'$ which $\tau$ satisfies.
      As $C$ is saturated we can deduce that $\psi' \cup \psi\ ?\ S \subseteq S_2 \in C$.
      Note that $\psi'$ and $\psi$ hold under $\tau$.

      If the constraint $S_2$ is in the domain of $\tau$, then $\tau S \subseteq \tau S_2$ by the definition of a partial solution, and so $k \in \tau S_2$ as required.

      If $S_2=\dom(Y(d))$ for some $Y \not\in \dom(\tau)$, however, by the definition of $\theta_\tau$, $k \in \dom((\theta_\tau Y)(d))$.
  \end{itemize}
\end{proof}

\end{toappendix}

\noindent
If there are no trivially unsatisfiable constraints in $\Sat(C)$, then a solution can be constructed as follows.
  For each variable $X$ occurring in $C$, define a function $\theta X$ as follows:
  \[
    (\theta X)(\underline{d}) \coloneqq \{k \mid k \in X(d)\ \text{is in}\ \Sat(C) \}
  \] 
  Then $\theta$ solves $\Sat(C)$.
 Conversely, if there is a trivially unsatisfiable constraint in $\Sat(C)$, then $\Sat(C)$ is unsolvable and, by the equivalence theorem, it follows that $C$ has no solution either.

\begin{theoremrep}
  $C$ is unsatisfiable iff $\Sat(C)$ is trivially unsatisfiable.
\end{theoremrep}
\begin{proof}
  The consistency of $\Sat(C)$ follows from Lemma \ref{lem:sat-consistency} with partial solution $\tau$ being empty and then Theorem \ref{thm:equi-consistency} gives the consistency of $C$.

  In the backward direction, Theorem \ref{thm:equi-consistency} gives consistency of $\Sat(C)$.  Since no substitution can solve $k \in \emptyset$, it follows that therefore $k \in \emptyset \notin \Sat(C)$.
\end{proof}




\section{Restriction and Complexity}\label{sec:analysis}

\newcommand{\bigoh}{\mathcal{O}}

In practice, having established that a constraint set is solvable, we are only interested in the solutions for a certain subset of the refinement variables.
For example if, as we have seen, the constraints $C$ describe a set of types $\{T\theta \mid \theta\ \text{solves}\ C\}$ of the module level functions, then we may consider two solutions $\theta_1$ and $\theta_2$ to be the same whenever they agree on $\frv(T)$.
We call the free refinement variables of $T$, in this case, the \emph{interface variables}.

\begin{definition}
  Let $C$ be a saturated constraint set and let $I$ be some set of refinement variables, called the \emph{interface variables}.
  Then define the \emph{restriction of $C$ to $I$}, written $\restr{C}{I}$, as the set $\{\phi\ ?\ S_1 \subseteq S_2 \in C \mid \frv(\phi) \cup \frv(S_1) \cup \frv(S_2) \subseteq I \}$.
\end{definition}

The restriction of $C$ to $I$ is quite severe, since it simply \emph{discards} any constraint not solely comprised of interface variables.
However, a remarkably strong property of the rules in Figure \ref{fig:resolution} is that, whenever $C$ is solvable, \emph{every} solution of $\restr{\Sat(C)}{I}$ may be extended\footnote{In the sense of Theorem \ref{thm:restr-ext}.} to a solution of $\Sat(C)$ (and therefore of $C$) --- independent of the choice of $I$!  
Since every solution of $\Sat(C)$ trivially restricts to a solution of $\restr{\Sat(C)}{I}$ (the latter has fewer constraints over fewer variables), it follows that the solutions of $\restr{\Sat(C)}{I}$ are \emph{exactly} the restrictions of the solutions of $C$.


\begin{example}\label{ex:sat}
Consider the following constraint set $C$ by way of an illustration:
\[
  \begin{array}{rlcrl}
    \star&\mathsf{Cst} \in X_1(\mathsf{Lam}) &\phantom{woo}&
    &X_1(\mathsf{Lam}) \subseteq X_2(\mathsf{Lam}) \\
    &\mathsf{FVr} \in X_2(\mathsf{Lam})\ ?\ X_2(\mathsf{Lam}) \subseteq X_3(\mathsf{Lam}) &\phantom{woo}&
    \star&X_3(\mathsf{Lam}) \subseteq \{\mathsf{FVr},\,\mathsf{Cst}\}
  \end{array}
\]
This set is not saturated and, consequently, there is no guarantee that the restriction of this set to an interface results in a constraint system whose solutions can generally be extended to solutions of the original set $C$.  For example, if we restrict this set to the interface $I = \{X_1,X_3\}$, the effect will be to retain the two starred constraints.  But then 
\[
  \theta(X)(\ul{d}) \coloneqq
    \begin{cases}
      \{ \mathsf{Cst},\, \mathsf{FVr},\, \mathsf{App} \} & \text{if $X = X_1$ and $\ul{d} = \mathsf{Lam}$} \\
      \emptyset & \text{otherwise}
    \end{cases}
\]
is a solution of $\restr{C}{I}$ that does not extend to any solution of $C$.
However, after saturation, $\Sat(C)$ consists of the following:
\[
  \small
  \begin{array}{rlcrl}
   \star &\mathsf{Cst} \in X_1(\mathsf{Lam}) &\phantom{woo}&
    &X_1(\mathsf{Lam}) \subseteq X_2(\mathsf{Lam}) \\
    &\mathsf{Cst} \in X_2(\mathsf{Lam}) &\phantom{woo}&
    &\mathsf{FVr} \in X_2(\mathsf{Lam})\ ?\ X_2(\mathsf{Lam}) \subseteq X_3(\mathsf{Lam}) \\
    &\mathsf{FVr} \in X_2(\mathsf{Lam})\ ?\ X_1(\mathsf{Lam}) \subseteq X_3(\mathsf{Lam}) &\phantom{woo}&
    &\mathsf{FVr} \in X_2(\mathsf{Lam})\ ?\ \mathsf{Cst} \in X_3(\mathsf{Lam}) \\
    &\mathsf{FVr} \in X_1(\mathsf{Lam})\ ?\ X_2(\mathsf{Lam}) \subseteq X_3(\mathsf{Lam}) &\phantom{woo}&
   \star &\mathsf{FVr} \in X_1(\mathsf{Lam})\ ?\ X_1(\mathsf{Lam}) \subseteq X_3(\mathsf{Lam}) \\
   \star &\mathsf{FVr} \in X_1(\mathsf{Lam})\ ?\ \mathsf{Cst} \in X_3(\mathsf{Lam}) &\phantom{woo}&
   \star &X_3(\mathsf{Lam}) \subseteq \{\mathsf{FVr},\,\mathsf{Cst}\} \\
    &\mathsf{FVr} \in X_2(\mathsf{Lam})\ ?\ X_2(\mathsf{Lam}) \subseteq \{\mathsf{FVr},\,\mathsf{Cst}\} &\phantom{woo}&
    &\mathsf{FVr} \in X_2(\mathsf{Lam})\ ?\ X_1(\mathsf{Lam}) \subseteq \{\mathsf{FVr},\,\mathsf{Cst}\} \\
    &\mathsf{FVr} \in X_1(\mathsf{Lam})\ ?\ X_2(\mathsf{Lam}) \subseteq \{\mathsf{FVr},\,\mathsf{Cst}\} &\phantom{woo}&
   \star &\mathsf{FVr} \in X_1(\mathsf{Lam})\ ?\ X_1(\mathsf{Lam}) \subseteq \{\mathsf{FVr},\,\mathsf{Cst} \}
  \end{array}
\]
In particular, the constraint $\mathsf{FVr} \in X_1(Lam)\ ?\ X_1(Lam) \subseteq X_3(Lam)$ which will be retained in the restriction $\restr{\Sat(C)}{I}$, which consists of the starred constraints.  
Consequently, the above assignment $\theta$ is ruled out.
Indeed, one can easily verify that every solution of $\restr{\Sat(C)}{I}$ extends to a solution of $\Sat(C)$ and hence of $C$.
\end{example}

\begin{theorem}[Restriction/Extension]\label{thm:restr-ext}
  Suppose $C$ is a saturated constraint set and $I$ is a subset of variables.  Let $\theta$ be a solution for $\restr{C}{I}$.  Then there is a solution $\theta'$ for $C$ satisfying, for all $X \in I$:
  \[
    \theta'(X)(\ul{d}) = \theta(X)(\ul{d})
  \]
\end{theorem}


Although our inference procedure is compositional, i.e. it breaks modules down into top-level definitions, and terms down into sub-terms that can be analysed in isolation, this is no guarantee of its efficiency.
As we have described it in Section~\ref{sec:inference}, the number of constraints associated with a function definition depends on the size of the definition --- constraints are generated at most syntax nodes and propagated to the root.
In fact, as is well known for constrained type inference, the situation is worse than simply this, because a whole set of constraints is imported from the environment when inferring for a program variable $x$.
The number of constraints associated with $x$ will again depend on the size of the definition of $x$ and the number of constraints associated with any functions that $x$ depends on, and so the number of constraints can become exponential in the number of function definitions\footnote{Although it is known that this can be avoided by a clever  representation in the case of constraints that are only simple variable/variable inclusions \cite{RN52}.}.

Let us fix $N$ to be the number of module-level function definitions, $K$ the maximum number of constructors associated with any datatype and $D$ the maximum number of datatypes associated with any slice (for \textsf{Lam}, this is 2).
A simple analysis of the shape of constraints yields the bound:

\begin{toappendix}
  Let $S$ be the size of the largest function definition.
  Let $K$ be the maximum number of constructors associated with any datatype, let $D$ be the maximum number of datatypes in any slice (i.e. for \textsf{Lam} this is 2) and let $Q$ be the maximum size of any underlying type.
\end{toappendix}

\begin{lemmarep}\label{lem:max-constraints}
  There are $\bigoh(Kv^2D \cdot 2^{vKD+K})$ atomic constraints over $v$ refinement variables.
\end{lemmarep}
\begin{proof}
  Each non-trivially-unsatisfiable atomic constraint $\phi\ ?\ S_1 \subseteq S_2$ can be understood as a choice one of $2^K$ subsets of constructors for each of the $vD$ pairs of refinement variable and underlying datatype (for $\phi$, we assume that each variable is associated with a particular slice), followed by a choice of one of $vD$ variable and datatype pairs for the head, which will appear in either the $S_1$ or $S_2$ position depending on the next choice, and then either: another variable (the underlying datatype is determined by previous choice) for the other position, one of $K$ constructors (for $S_1$) or a choice of the $2^K$ subsets of constructors or a variable datatype pair (for $S_2$).
  Therefore, there are at most $2^{KvD} \cdot vD \cdot (v + K + 2^K)$ possible constraints over these refinement variables and:
  \[
    2^{KvD} \cdot vD \cdot (v + K + 2^K) = v^2D2^{KvD} + KvD2^{KvD} + vD2^{KvD+K)} = \bigoh(Kv^2D \cdot 2^{KvD + K})
  \]
\end{proof}

Suppose $\Gamma \types e \infers T,\,C$.  As it stands, the number of variables $v$ occurring in $C$ will depend upon the size of $e$ and the size of every definition that $e$ depends on.
We use restriction to break this dependency between the size of constraint sets and the size of the program.  
We compute $\restr{C}{I}$ for each constraint set $C$ generated by an inference (i.e. at every step) with the interface $I$ taken to be the free variables of the context $ \frv(\Gamma) \cup \frv(T)$.
Consequently, the number of variables $v$ occurring in $\restr{C}{I}$ only depends on the number of refinement variables that are free in $\Gamma$ and $T$.
Since all the refinement variables of module level function types are generalised by \rlnm{IModD} (assuming, as is usual, that inference for modules occurs in a closed environment), it follows that the only free refinement variables in $\Gamma$ are those introduced during the inference of $e$, as a result of inferring under abstractions and case expressions. 
Thus $v$ becomes independent of the number of function definitions.

Moreover, if we assume that function definitions are in $\beta$-normal form and that the maximum case expression nesting depth within any given definition is fixed (i.e. does not grow with the size of the program) then it follows that the number of refinement variables free in the environment is bounded by the size of the underlying type of $e$.
Clearly, the number of free refinement variables in $T$ is also bounded by its underlying type.

Consequently, for a constraint set $C$ arising by an inference $\Gamma \types e \infers T,\,C$ and then restricted to its context, the number of constraints given by Lemma~\ref{lem:max-constraints} only depends on the size of the underlying types assigned to the $e$ and, in the case of datatypes, the size of their definitions (slices).  
If we consider scaling our analysis to larger and larger programs to mean programs consisting of more and more functions, with bounded growth in the size of types and the size of individual function definitions then we may reasonably consider all these parameters fixed.  Recall that $N$ is the number of function definitions in the module. We have:

\begin{toappendix}

Suppose the maximum level of \textsf{case} expression nesting is $M$ and let $V = M + 2Q$.

\begin{lemma}\label{lem:max-inferred}
  Suppose $\Gamma \types e \infers T,\,C$ with $C$ restricted to the context. The  number of constraints in $C$ is $\bigoh(KV^2D \cdot 2^{KVD+K})$.
\end{lemma}
\begin{proof}
  By applying restriction after each inference, we are guaranteed that the refinement variables occurring in $C$ are a subset of those occurring in $\Gamma$ and $T$.  The only free refinement variables in $\Gamma$ are those that were introduced as a result of inferring under a lambda abstraction (i.e. introduced by the \rlnm{IAbs} rule) or those introduced as a result of inferring under a case statement (i.e. introduced by the \rlnm{ICase} rule).  There are at most $Q$ introduced by abstractions and a further $M$ introduced by case matching (each case only introduces a single refinement variable $X$, independently of the complexity of the datatype of the scrutinees).  
  The type $T$ introduces at most another $Q$ many.  
  
\end{proof}

\begin{lemma}\label{lem:inf-complexity}
  The complexity of type inference is $\bigoh(N S K^5V^4D^2 \cdot 2^{K(VD(K+1)+2)})$
\end{lemma}
\begin{proof}
  For each subterm of each function definition, inference must generate new constraints and apply the saturation and restriction operations to the set of constraints obtained from combining the outputs of its recursive calls.
  Using a standard fixpoint computation, and the number of possible constraints over a given set of variables calculated in Lemma \ref{lem:max-constraints}, saturation of a constraint set of size $n$ involving $v$ refinement variables can be achieved in time $\bigoh(nKv^2D \cdot 2^{KvD+K})$.  Hence the time taken to process each subterm will be dominated by the time by saturation.  The number of constraints that constitutes the input to saturation for a given subterm is a function of the inference rule applied, its premises and side conditions.  The worst-case is \rlnm{ICase} which has at most $K$ premises and a number of side conditions.  The total number of constraints $C$ before saturation and restriction, in this case, is the sum of the sizes of each $C_i$, $C_i'$, $C_0$ and a single constraint contributed by the side condition.  Each of these subsets is already restricted to its context so, by Lemma~\ref{lem:max-inferred}, their sizes are at most $\bigoh(KV^2D \cdot 2^{KVD+K})$, giving an overall upper bound for $C$ of $\bigoh(K^2V^2D \cdot 2^{KVD+K})$.
  The total number of refinement variables occurring in this set is given by those free in the $\Gamma$, $T$ and each $T_i$.  The context $\Gamma$ and type $T$ contribute at most $V$ many and the $T_i$ together contribute at most a further $KQ$.  Consequently, the number is bounded above by $KV$.  Hence, saturation can be computed in time $\bigoh(K^5V^4D^2 \cdot 2^{K(VD(K+1)+2)})$.
\end{proof}

\end{toappendix}

\begin{theoremrep}\label{thm:inf-complexity}
  Under the assumption that the size of types and the size of individual function definitions is bounded, the complexity of type inference is $\bigoh(N)$.
\end{theoremrep}
\begin{proof}
  Follows immediately from Lemma~\ref{lem:inf-complexity} by fixing all other parameters.
\end{proof}



\section{Implementation}\label{sec:implementation}
We implemented a prototype of our inference algorithm for Haskell as a GHC plugin.
The user can run our type checker as another stage of compilation with an additional command line flag.
In addition to running the type checker on individual modules, an interface binary file is generated, enabling other modules to use the constraint information in separate compilations.
It is available from: \url{https://github.com/bristolpl/intensional-datatys}.

Our plugin processes GHC's core language \cite{sulzmann-et-al-2007}, which is significantly more powerful than the small language presented here.
Specifically, it must account for higher-rank types (including existentials), casts and coercions, type classes.
We have not implemented a treatment of these features in our prototype and so any occurrences are not analysed.
Furthermore, we disallow empty refinements of single-constructor datatypes (e.g. records).
This relatively small departure from the theory is a substantial improvement to the efficiency of the tool due to the number of records and newtypes that are found in typical Haskell programs.

Since we do not analyse the dependencies of packages, datatypes that are defined outside the current package are treated as base types and not refined
The resulting analysis provides a certificate of safety for some package modulo the safe use of its dependencies.

\clearpage
In addition to missing cases, the tool uses the results of internal analyses in GHC to identify pattern matching cases that will throw an exception.  For example, the following code will be considered as potentially unsafe.
\begin{lstlisting}
      nnf2dnf (Lit a) = [[a]]
      nnf2dnf (Or p q) = List.union (nnf2dnf p) (nnf2dnf q)
      nnf2dnf (And p q) = distrib (nnf2dnf p) (nnf2dnf q)
      nnf2dnf _ = error "Impossible case!"
\end{lstlisting}


\subsection{Performance}
We recorded benchmarks on a 2.20GHz Intel\textregistered Core\texttrademark i5-5200U with 4 cores and 8.00GB RAM.
We used the following selection of projects from the Hackage database:

\begin{itemize}
  \item \texttt{aeson} is a performant JSON serialisation library.
  \item The \texttt{containers} package provides a selection of classic functional data structures such as sets and finite maps.
        The Data.Sequence module from this package contains machine generated code that lacks the typical modularity and structure of hand written code.
        For example, it contains an automatically generated set of 6 mutually recursive functions\footnote{Since they are mutually recursive, they are processed together before generalisation and thus act as a single complex type.}, each with a complex type and deeply nested matching.
        The corresponding interface is in excess of 80 refinement variables.
        This module could not be processed to completion in a small amount of time and so we have omitted it from the results.
        We will explore how best to process examples that violate our complexity assumptions in follow-up work.
  \item \texttt{extra} is a collection of common combinators for datatypes and control flow.
  \item \texttt{fgl} (Functional Graph Library) provides an inductive representation of graphs.
  \item \texttt{haskeline} is a command-line interface library
  \item \texttt{parallel} is Haskell's default library for parallel programming
  \item \texttt{sbv} is an SMT based automatic verification tool for Haskell programs.
  \item The \texttt{time} library contains several representations of time, clocks and calendars.
  \item \texttt{unordered-containers} provides hashing-based containers, for either performant code or datatypes without a natural ordering.
\end{itemize}

For each module we recorded the average time elapsed in milliseconds across 10 runs and the number of top-level definitions (N).
We note both the total number of refinement variables generated during inference (V) and the largest interface (I).
The contrast between these two figures gives some indication of how intractable the analysis may become be without the restriction operator.
Naturally, constant factors will vary considerably between modules (not in correspondence with their size) and so our results also include the number of constructors (K) that appear in the largest datatype, and the number of datatypes (D) in the largest slice.

The benchmarks in Figure \ref{tvl:bench-summaries} provide a summary of the results for each project, i.e. the total time taken\footnote{The total time taken is the sum of the time taken to analyse each module independently doesn't include start up costs etc.}
, the total number of top-level definitions, the total number of refinement variables, the maximum interface size, the largest number of constructors associated with a datatype, and the largest slice.
The full dataset can be found in the appendices and a virtual machine image for recreating the benchmarks can be downloaded from: \url{https://doi.org/10.5281/zenodo.4072906}.

Figure \ref{tvl:bench-summaries} also contains the number of warnings found in each packages.
However, many of them stem from the same incomplete pattern.
For example, 70 of the warnings from the sbv package are located in one function.
All of these warnings were due to the tools limited, and thus extremely conservative, approach to handle features of GHC outside of its scope, such as typeclasses and encapulsation via the module system, so we are optimistic about future work.

These packages were selected to test the tool in a range of contexts and at scale.
We did not find any true positives, but it is not surprising since large packages with many downloads on Hackage are likely to be quite mature.

\begin{table}
  \caption{Benchmark Summaries}\label{tvl:bench-summaries}
  \begin{tabular}{|l|l|l|l|l|l|l|l|}
  \hline
  Name                 & N    & K  & V      & D  & I  & Warnings & Time (ms) \\\hline
  aeson                & 728  & 13 & 20466  & 6  & 14 & 0        & 79.37     \\\hline
  containers           & 1792 & 5  & 25237  & 2  & 23 & 18       & 118.26    \\\hline
  extra                & 332  & 3  & 5438   & 3  & 7  & 0        & 61.53     \\\hline
  fgl                  & 700  & 2  & 18403  & 2  & 12 & 8        & 94.32     \\\hline
  haskeline            & 1384 & 15 & 29389  & 19 & 27 & 0        & 111.67    \\\hline
  parallel             & 110  & 1  & 959    & 2  & 18 & 0        & 10.18     \\\hline
  pretty               & 222  & 8  & 3675   & 4  & 16 & 11       & 23.86     \\\hline
  sbv                  & 5076 & 44 & 171869 & 49 & 46 & 79       & 518.91    \\\hline
  time                 & 484  & 7  & 9753   & 6  & 10 & 9        & 134.16    \\\hline
  unordered-containers & 474  & 5  & 7761   & 3  & 24 & 2        & 30.56     \\\hline
\end{tabular}
\end{table}

\begin{toappendix}

  The following is a listing of benchmark results with packages split by module.

  \subsection{Package \texttt{aeson-1.5.2.0}}

  \begin{tabular}{|l|l|l|l|l|l|l|}
\hline
Name & N & K & V & D & I & Time (ms)\\\hline
Data.Aeson & 15 & 6 & 198 & 2 & 3 & 3.01\\\hline
Data.Aeson.Encode & 2 & 0 & 1 & 0 & 1 & 2.75\\\hline
Data.Aeson.Encoding & 0 & 0 & 0 & 0 & 0 & 2.72\\\hline
Data.Aeson.Encoding.Builder & 71 & 6 & 1534 & 1 & 6 & 3.34\\\hline
Data.Aeson.Encoding.Internal & 67 & 6 & 723 & 2 & 7 & 2.65\\\hline
Data.Aeson.Internal & 0 & 0 & 0 & 0 & 0 & 2.85\\\hline
Data.Aeson.Internal.Functions & 3 & 0 & 10 & 0 & 0 & 3.33\\\hline
Data.Aeson.Internal.Time & 0 & 0 & 0 & 0 & 0 & 3.26\\\hline
Data.Aeson.Parser & 0 & 0 & 0 & 0 & 0 & 2.62\\\hline
Data.Aeson.Parser.Internal & 63 & 6 & 2252 & 2 & 9 & 3.20\\\hline
Data.Aeson.Parser.Time & 7 & 1 & 35 & 1 & 1 & 2.81\\\hline
Data.Aeson.Parser.Unescape & 0 & 0 & 0 & 0 & 0 & 3.01\\\hline
Data.Aeson.Parser.UnescapePure & 19 & 13 & 1109 & 2 & 5 & 4.09\\\hline
Data.Aeson.QQ.Simple & 2 & 6 & 111 & 1 & 2 & 7.12\\\hline
Data.Aeson.TH & 271 & 4 & 10251 & 3 & 11 & 2.88\\\hline
Data.Aeson.Text & 13 & 6 & 326 & 1 & 4 & 3.99\\\hline
Data.Aeson.Types & 1 & 1 & 7 & 1 & 2 & 2.81\\\hline
Data.Aeson.Types.Class & 0 & 0 & 0 & 0 & 0 & 2.76\\\hline
Data.Aeson.Types.FromJSON & 96 & 6 & 2054 & 6 & 14 & 2.93\\\hline
Data.Aeson.Types.Generic & 1 & 0 & 1 & 0 & 1 & 3.77\\\hline
Data.Aeson.Types.Internal & 46 & 6 & 564 & 2 & 8 & 3.64\\\hline
Data.Aeson.Types.ToJSON & 27 & 6 & 484 & 3 & 10 & 2.88\\\hline
Data.Attoparsec.Time & 16 & 1 & 745 & 1 & 2 & 2.94\\\hline
Data.Attoparsec.Time.Internal & 8 & 1 & 61 & 1 & 1 & 3.98\\\hline
\end{tabular}

  \clearpage

  \subsection{Package \texttt{containers-0.6.2.1}}

  \begin{tabular}{|l|l|l|l|l|l|l|}
\hline
Name & N & K & V & D & I & Time (ms)\\\hline
Data.Containers.ListUtils & 11 & 3 & 61 & 1 & 4 & 3.97\\\hline
Data.Graph & 69 & 2 & 1411 & 2 & 8 & 2.87\\\hline
Data.IntMap & 8 & 0 & 162 & 0 & 2 & 2.83\\\hline
Data.IntMap.Internal & 438 & 3 & 5519 & 2 & 23 & 2.62\\\hline
Data.IntMap.Internal.Debug & 0 & 0 & 0 & 0 & 0 & 2.43\\\hline
Data.IntMap.Internal.DeprecatedDebug & 4 & 0 & 84 & 0 & 1 & 2.95\\\hline
Data.IntMap.Lazy & 0 & 0 & 0 & 0 & 0 & 2.62\\\hline
Data.IntMap.Merge.Lazy & 0 & 0 & 0 & 0 & 0 & 15.55\\\hline
Data.IntMap.Merge.Strict & 15 & 3 & 152 & 2 & 4 & 2.71\\\hline
Data.IntMap.Strict & 0 & 0 & 0 & 0 & 0 & 3.08\\\hline
Data.IntMap.Strict.Internal & 146 & 3 & 1541 & 2 & 14 & 9.76\\\hline
Data.IntSet & 0 & 0 & 0 & 0 & 0 & 2.55\\\hline
Data.IntSet.Internal & 288 & 5 & 3690 & 2 & 9 & 2.89\\\hline
Data.Map & 10 & 0 & 200 & 0 & 2 & 2.86\\\hline
Data.Map.Internal & 379 & 4 & 5697 & 2 & 15 & 2.85\\\hline
Data.Map.Internal.Debug & 20 & 2 & 538 & 1 & 3 & 3.33\\\hline
Data.Map.Internal.DeprecatedShowTree & 4 & 0 & 86 & 0 & 1 & 2.46\\\hline
Data.Map.Lazy & 0 & 0 & 0 & 0 & 0 & 2.57\\\hline
Data.Map.Merge.Lazy & 0 & 0 & 0 & 0 & 0 & 2.72\\\hline
Data.Map.Merge.Strict & 0 & 0 & 0 & 0 & 0 & 1.77\\\hline
Data.Map.Strict & 0 & 0 & 0 & 0 & 0 & 3.07\\\hline
Data.Map.Strict.Internal & 137 & 2 & 2013 & 2 & 12 & 2.81\\\hline
Data.Set & 0 & 0 & 0 & 0 & 0 & 2.75\\\hline
Data.Set.Internal & 211 & 2 & 3362 & 2 & 11 & 2.86\\\hline
Data.Tree & 24 & 1 & 497 & 1 & 7 & 2.84\\\hline
Utils.Containers.Internal.BitQueue & 13 & 1 & 153 & 2 & 4 & 4.02\\\hline
Utils.Containers.Internal.BitUtil & 5 & 0 & 22 & 0 & 0 & 3.40\\\hline
Utils.Containers.Internal.Coercions & 2 & 0 & 6 & 0 & 0 & 4.58\\\hline
Utils.Containers.Internal.PtrEquality & 2 & 0 & 19 & 0 & 0 & 3.04\\\hline
Utils.Containers.Internal.State & 2 & 1 & 5 & 1 & 1 & 3.69\\\hline
Utils.Containers.Internal.StrictMaybe & 3 & 2 & 15 & 1 & 1 & 3.69\\\hline
Utils.Containers.Internal.StrictPair & 1 & 0 & 4 & 0 & 1 & 3.40\\\hline
Utils.Containers.Internal.TypeError & 0 & 0 & 0 & 0 & 0 & 2.74\\\hline
\end{tabular}

  \clearpage

  \subsection{Package \texttt{extra-1.7.3}}

  \begin{tabular}{|l|l|l|l|l|l|l|}
\hline
Name & N & K & V & D & I & Time (ms)\\\hline
Control.Concurrent.Extra & 20 & 3 & 552 & 3 & 7 & 2.41\\\hline
Control.Exception.Extra & 19 & 0 & 348 & 0 & 0 & 2.85\\\hline
Control.Monad.Extra & 57 & 0 & 342 & 0 & 0 & 2.93\\\hline
Data.Either.Extra & 9 & 0 & 94 & 0 & 0 & 2.91\\\hline
Data.IORef.Extra & 4 & 0 & 40 & 0 & 0 & 3.50\\\hline
Data.List.Extra & 119 & 3 & 1963 & 1 & 6 & 3.47\\\hline
Data.List.NonEmpty.Extra & 17 & 0 & 84 & 0 & 0 & 2.36\\\hline
Data.Tuple.Extra & 16 & 0 & 58 & 0 & 0 & 4.71\\\hline
Data.Typeable.Extra & 0 & 0 & 0 & 0 & 0 & 2.76\\\hline
Data.Version.Extra & 4 & 0 & 122 & 0 & 0 & 2.63\\\hline
Extra & 0 & 0 & 0 & 0 & 0 & 3.26\\\hline
Numeric.Extra & 5 & 0 & 19 & 0 & 0 & 3.46\\\hline
Partial & 0 & 0 & 0 & 0 & 0 & 3.23\\\hline
System.Directory.Extra & 8 & 0 & 265 & 0 & 0 & 2.72\\\hline
System.Environment.Extra & 0 & 0 & 0 & 0 & 0 & 3.55\\\hline
System.IO.Extra & 35 & 0 & 923 & 0 & 0 & 2.47\\\hline
System.Info.Extra & 2 & 0 & 2 & 0 & 0 & 2.70\\\hline
System.Process.Extra & 5 & 0 & 264 & 0 & 0 & 4.35\\\hline
System.Time.Extra & 12 & 1 & 362 & 1 & 3 & 2.29\\\hline
Text.Read.Extra & 0 & 0 & 0 & 0 & 0 & 2.98\\\hline
\end{tabular}

  \clearpage

  \subsection{Package \texttt{fgl-5.7.0.2}}

  \begin{tabular}{|l|l|l|l|l|l|l|}
\hline
Name & N & K & V & D & I & Time (ms)\\\hline
Data.Graph.Inductive & 1 & 0 & 20 & 0 & 0 & 2.95\\\hline
Data.Graph.Inductive.Basic & 17 & 0 & 452 & 0 & 0 & 2.47\\\hline
Data.Graph.Inductive.Example & 53 & 1 & 5236 & 1 & 1 & 2.81\\\hline
Data.Graph.Inductive.Graph & 102 & 1 & 1920 & 1 & 1 & 4.26\\\hline
Data.Graph.Inductive.Internal.Heap & 19 & 2 & 314 & 1 & 6 & 3.71\\\hline
Data.Graph.Inductive.Internal.Queue & 5 & 1 & 44 & 1 & 4 & 3.40\\\hline
Data.Graph.Inductive.Internal.RootPath & 6 & 1 & 123 & 1 & 2 & 3.22\\\hline
Data.Graph.Inductive.Internal.Thread & 20 & 0 & 133 & 0 & 0 & 3.66\\\hline
Data.Graph.Inductive.Monad & 23 & 0 & 349 & 0 & 0 & 3.31\\\hline
Data.Graph.Inductive.Monad.IOArray & 7 & 1 & 272 & 1 & 1 & 2.47\\\hline
Data.Graph.Inductive.Monad.STArray & 8 & 1 & 282 & 1 & 1 & 2.93\\\hline
Data.Graph.Inductive.NodeMap & 60 & 1 & 678 & 1 & 4 & 3.81\\\hline
Data.Graph.Inductive.PatriciaTree & 23 & 1 & 871 & 1 & 2 & 2.79\\\hline
Data.Graph.Inductive.Query & 0 & 0 & 0 & 0 & 0 & 2.74\\\hline
Data.Graph.Inductive.Query.ArtPoint & 16 & 1 & 442 & 1 & 5 & 3.40\\\hline
Data.Graph.Inductive.Query.BCC & 20 & 0 & 500 & 0 & 0 & 3.18\\\hline
Data.Graph.Inductive.Query.BFS & 29 & 1 & 625 & 2 & 9 & 3.10\\\hline
Data.Graph.Inductive.Query.DFS & 44 & 0 & 467 & 0 & 0 & 3.14\\\hline
Data.Graph.Inductive.Query.Dominators & 36 & 0 & 655 & 0 & 0 & 4.05\\\hline
Data.Graph.Inductive.Query.GVD & 10 & 2 & 209 & 2 & 3 & 2.88\\\hline
Data.Graph.Inductive.Query.Indep & 9 & 0 & 208 & 0 & 0 & 4.08\\\hline
Data.Graph.Inductive.Query.MST & 13 & 2 & 219 & 2 & 10 & 3.76\\\hline
Data.Graph.Inductive.Query.MaxFlow & 20 & 0 & 322 & 0 & 1 & 2.63\\\hline
Data.Graph.Inductive.Query.MaxFlow2 & 51 & 2 & 2500 & 2 & 12 & 2.89\\\hline
Data.Graph.Inductive.Query.Monad & 58 & 1 & 700 & 1 & 6 & 2.87\\\hline
Data.Graph.Inductive.Query.SP & 12 & 2 & 177 & 2 & 10 & 3.19\\\hline
Data.Graph.Inductive.Query.TransClos & 15 & 0 & 215 & 0 & 0 & 3.62\\\hline
Data.Graph.Inductive.Tree & 8 & 1 & 306 & 1 & 1 & 3.87\\\hline
Paths\_fgl & 15 & 0 & 164 & 0 & 0 & 3.13\\\hline
\end{tabular}

  \clearpage

  \subsection{Package \texttt{haskeline-0.8.0.1}}

  \begin{tabular}{|l|l|l|l|l|l|l|}
\hline
Name & N & K & V & D & I & Time (ms)\\\hline
System.Console.Haskeline & 136 & 15 & 2343 & 17 & 18 & 3.55\\\hline
System.Console.Haskeline.Backend & 6 & 15 & 55 & 7 & 2 & 4.28\\\hline
System.Console.Haskeline.Backend.DumbTerm & 48 & 15 & 484 & 7 & 12 & 3.65\\\hline
System.Console.Haskeline.Backend.Posix & 51 & 15 & 1843 & 8 & 10 & 3.96\\\hline
System.Console.Haskeline.Backend.Posix.Encoder & 8 & 2 & 98 & 2 & 2 & 4.28\\\hline
System.Console.Haskeline.Backend.Terminfo & 148 & 15 & 2108 & 7 & 13 & 3.98\\\hline
System.Console.Haskeline.Backend.WCWidth & 10 & 1 & 168 & 1 & 6 & 4.25\\\hline
System.Console.Haskeline.Command & 28 & 15 & 327 & 8 & 7 & 4.09\\\hline
System.Console.Haskeline.Command.Completion & 47 & 15 & 807 & 13 & 13 & 3.68\\\hline
System.Console.Haskeline.Command.History & 52 & 15 & 1131 & 11 & 12 & 2.90\\\hline
System.Console.Haskeline.Command.KillRing & 35 & 15 & 360 & 9 & 6 & 3.59\\\hline
System.Console.Haskeline.Command.Undo & 17 & 15 & 138 & 8 & 3 & 3.91\\\hline
System.Console.Haskeline.Completion & 40 & 1 & 948 & 1 & 6 & 3.37\\\hline
System.Console.Haskeline.Directory & 0 & 0 & 0 & 0 & 0 & 3.95\\\hline
System.Console.Haskeline.Emacs & 73 & 15 & 2328 & 9 & 21 & 3.55\\\hline
System.Console.Haskeline.History & 16 & 1 & 406 & 1 & 3 & 6.40\\\hline
System.Console.Haskeline.IO & 11 & 15 & 185 & 19 & 6 & 4.53\\\hline
System.Console.Haskeline.InputT & 65 & 15 & 1999 & 17 & 18 & 3.35\\\hline
System.Console.Haskeline.Internal & 16 & 15 & 417 & 17 & 16 & 4.86\\\hline
System.Console.Haskeline.Key & 26 & 15 & 647 & 3 & 2 & 5.61\\\hline
System.Console.Haskeline.LineState & 67 & 2 & 710 & 3 & 6 & 5.61\\\hline
System.Console.Haskeline.Monads & 12 & 0 & 21 & 0 & 0 & 4.84\\\hline
System.Console.Haskeline.Prefs & 26 & 15 & 769 & 8 & 5 & 4.35\\\hline
System.Console.Haskeline.Recover & 2 & 0 & 98 & 0 & 0 & 3.98\\\hline
System.Console.Haskeline.RunCommand & 33 & 15 & 582 & 8 & 27 & 3.81\\\hline
System.Console.Haskeline.Term & 44 & 15 & 628 & 7 & 4 & 3.82\\\hline
System.Console.Haskeline.Vi & 367 & 15 & 9789 & 12 & 27 & 3.50\\\hline
\end{tabular}

  \clearpage

  \subsection{Package \texttt{parallel-3.2.2.0}}

  \begin{tabular}{|l|l|l|l|l|l|l|}
\hline
Name & N & K & V & D & I & Time (ms)\\\hline
Control.Parallel & 2 & 0 & 0 & 0 & 0 & 2.98\\\hline
Control.Parallel.Strategies & 87 & 1 & 769 & 2 & 18 & 3.47\\\hline
Control.Seq & 21 & 0 & 190 & 0 & 0 & 3.74\\\hline
\end{tabular}

  \subsection{Package \texttt{pretty-1.1.3.6}}

  \begin{tabular}{|l|l|l|l|l|l|l|}
\hline
Name & N & K & V & D & I & Time (ms)\\\hline
Text.PrettyPrint & 0 & 0 & 0 & 0 & 0 & 3.09\\\hline
Text.PrettyPrint.Annotated & 0 & 0 & 0 & 0 & 0 & 3.91\\\hline
Text.PrettyPrint.Annotated.HughesPJ & 154 & 8 & 3148 & 3 & 16 & 4.23\\\hline
Text.PrettyPrint.Annotated.HughesPJClass & 5 & 8 & 32 & 3 & 3 & 4.29\\\hline
Text.PrettyPrint.HughesPJ & 58 & 8 & 472 & 4 & 7 & 5.23\\\hline
Text.PrettyPrint.HughesPJClass & 5 & 8 & 23 & 4 & 3 & 3.11\\\hline
\end{tabular}

  \subsection{Package \texttt{sbv-8.7.5}}
  \small
  \begin{tabular}{|l|l|l|l|l|l|l|}
\hline
Name & N & K & V & D & I & Time (ms)\\\hline
Data.SBV & 0 & 0 & 0 & 0 & 0 & 3.54\\\hline
Data.SBV.Char & 48 & 44 & 709 & 43 & 6 & 3.11\\\hline
Data.SBV.Client & 20 & 44 & 1007 & 43 & 5 & 3.38\\\hline
Data.SBV.Client.BaseIO & 118 & 0 & 418 & 0 & 4 & 3.27\\\hline
Data.SBV.Compilers.C & 262 & 44 & 12934 & 48 & 11 & 3.29\\\hline
Data.SBV.Compilers.CodeGen & 67 & 44 & 1919 & 49 & 8 & 3.07\\\hline
Data.SBV.Control & 1 & 2 & 5 & 1 & 3 & 3.75\\\hline
Data.SBV.Control.BaseIO & 50 & 0 & 142 & 0 & 5 & 2.87\\\hline
Data.SBV.Control.Query & 224 & 44 & 7932 & 43 & 20 & 3.52\\\hline
Data.SBV.Control.Types & 7 & 31 & 283 & 1 & 2 & 5.38\\\hline
Data.SBV.Control.Utils & 311 & 44 & 12922 & 44 & 33 & 4.02\\\hline
Data.SBV.Core.AlgReals & 49 & 2 & 1444 & 2 & 4 & 4.11\\\hline
Data.SBV.Core.Concrete & 55 & 14 & 2443 & 7 & 16 & 3.90\\\hline
Data.SBV.Core.Data & 47 & 44 & 261 & 43 & 6 & 3.46\\\hline
Data.SBV.Core.Floating & 69 & 44 & 1553 & 43 & 15 & 3.52\\\hline
Data.SBV.Core.Kind & 16 & 14 & 724 & 1 & 3 & 3.46\\\hline
Data.SBV.Core.Model & 229 & 44 & 3848 & 43 & 14 & 3.07\\\hline
Data.SBV.Core.Operations & 284 & 44 & 9524 & 43 & 37 & 5.50\\\hline
Data.SBV.Core.Sized & 53 & 44 & 701 & 43 & 7 & 3.72\\\hline
Data.SBV.Core.Symbolic & 235 & 44 & 6695 & 43 & 46 & 3.84\\\hline
Data.SBV.Dynamic & 17 & 44 & 274 & 44 & 9 & 3.59\\\hline
Data.SBV.Either & 48 & 44 & 864 & 43 & 13 & 3.10\\\hline
Data.SBV.Internals & 3 & 0 & 1 & 0 & 0 & 3.10\\\hline
Data.SBV.List & 80 & 44 & 1184 & 43 & 12 & 3.54\\\hline
Data.SBV.Maybe & 30 & 44 & 531 & 43 & 13 & 3.33\\\hline
Data.SBV.Provers.ABC & 1 & 31 & 57 & 9 & 3 & 3.74\\\hline
Data.SBV.Provers.Boolector & 1 & 31 & 71 & 9 & 3 & 3.58\\\hline
Data.SBV.Provers.CVC4 & 6 & 31 & 134 & 9 & 3 & 3.45\\\hline
Data.SBV.Provers.MathSAT & 3 & 31 & 101 & 9 & 3 & 3.31\\\hline
Data.SBV.Provers.Prover & 93 & 44 & 1374 & 48 & 11 & 3.46\\\hline
Data.SBV.Provers.Yices & 1 & 31 & 50 & 9 & 3 & 3.36\\\hline
Data.SBV.Provers.Z3 & 2 & 31 & 89 & 9 & 3 & 3.07\\\hline
Data.SBV.RegExp & 24 & 44 & 278 & 43 & 8 & 5.47\\\hline
Data.SBV.SMT.SMT & 118 & 44 & 5193 & 42 & 8 & 3.55\\\hline
Data.SBV.SMT.SMTLib & 9 & 44 & 785 & 14 & 43 & 3.33\\\hline
Data.SBV.SMT.SMTLib2 & 285 & 44 & 16859 & 13 & 36 & 3.32\\\hline
Data.SBV.SMT.SMTLibNames & 1 & 0 & 379 & 0 & 0 & 5.25\\\hline
Data.SBV.SMT.Utils & 33 & 31 & 563 & 9 & 2 & 3.31\\\hline
Data.SBV.Set & 76 & 44 & 1481 & 43 & 10 & 3.88\\\hline
Data.SBV.String & 53 & 44 & 1322 & 43 & 12 & 3.22\\\hline
\end{tabular}

\begin{tabular}{|l|l|l|l|l|l|l|}
    \hline
Data.SBV.Tools.BMC & 11 & 44 & 549 & 44 & 8 & 3.20\\\hline
Data.SBV.Tools.BoundedFix & 3 & 44 & 30 & 43 & 5 & 3.11\\\hline
Data.SBV.Tools.BoundedList & 59 & 44 & 578 & 43 & 12 & 19.58\\\hline
Data.SBV.Tools.CodeGen & 0 & 0 & 0 & 0 & 0 & 3.36\\\hline
Data.SBV.Tools.GenTest & 104 & 44 & 5552 & 43 & 10 & 3.03\\\hline
Data.SBV.Tools.Induction & 12 & 44 & 393 & 44 & 14 & 3.57\\\hline
Data.SBV.Tools.Overflow & 77 & 44 & 1491 & 43 & 12 & 3.14\\\hline
Data.SBV.Tools.Polynomial & 68 & 44 & 1586 & 43 & 12 & 3.11\\\hline
Data.SBV.Tools.Range & 39 & 44 & 951 & 45 & 13 & 3.29\\\hline
Data.SBV.Tools.STree & 19 & 44 & 488 & 44 & 20 & 3.39\\\hline
Data.SBV.Tools.WeakestPreconditions & 86 & 44 & 2829 & 45 & 33 & 3.68\\\hline
Data.SBV.Trans & 0 & 0 & 0 & 0 & 0 & 3.47\\\hline
Data.SBV.Trans.Control & 1 & 2 & 4 & 1 & 3 & 3.16\\\hline
Data.SBV.Tuple & 16 & 44 & 367 & 43 & 7 & 3.75\\\hline
Data.SBV.Utils.ExtractIO & 0 & 0 & 0 & 0 & 0 & 4.96\\\hline
Data.SBV.Utils.Lib & 35 & 3 & 1070 & 1 & 2 & 14.82\\\hline
Data.SBV.Utils.Numeric & 62 & 0 & 353 & 0 & 0 & 4.65\\\hline
Data.SBV.Utils.PrettyNum & 83 & 14 & 2387 & 6 & 9 & 3.78\\\hline
Data.SBV.Utils.SExpr & 115 & 6 & 7919 & 3 & 22 & 3.67\\\hline
Data.SBV.Utils.TDiff & 10 & 0 & 277 & 0 & 0 & 4.80\\\hline
Documentation.SBV.Examples.BitPrecise.BitTricks & 18 & 44 & 379 & 43 & 5 & 3.64\\\hline
Documentation.SBV.Examples.BitPrecise.BrokenSearch & 7 & 44 & 327 & 44 & 8 & 2.56\\\hline
Documentation.SBV.Examples.BitPrecise.Legato & 72 & 44 & 1347 & 49 & 24 & 3.30\\\hline
Documentation.SBV.Examples.BitPrecise.MergeSort & 19 & 44 & 428 & 49 & 9 & 3.05\\\hline
Documentation.SBV.Examples.BitPrecise.MultMask & 3 & 44 & 142 & 44 & 5 & 3.34\\\hline
Documentation.SBV.Examples.BitPrecise.PrefixSum & 17 & 44 & 269 & 44 & 4 & 3.57\\\hline
Documentation.SBV.Examples.CodeGeneration.AddSub & 3 & 44 & 96 & 49 & 5 & 3.15\\\hline
Documentation.SBV.Examples.CodeGeneration.CRC\_USB5 & 8 & 44 & 184 & 49 & 4 & 3.35\\\hline
Documentation.SBV.Examples.CodeGeneration.Fibonacci & 6 & 44 & 216 & 49 & 11 & 3.35\\\hline
Documentation.SBV.Examples.CodeGeneration.GCD & 7 & 44 & 192 & 49 & 11 & 3.66\\\hline
Documentation.SBV.Examples.CodeGeneration.PopulationCount & 7 & 44 & 191 & 49 & 8 & 3.24\\\hline
Documentation.SBV.Examples.CodeGeneration.Uninterpreted & 10 & 44 & 212 & 49 & 5 & 3.13\\\hline
Documentation.SBV.Examples.Crypto.AES & 216 & 44 & 5018 & 49 & 32 & 3.62\\\hline
Documentation.SBV.Examples.Crypto.RC4 & 22 & 44 & 793 & 44 & 18 & 3.45\\\hline
Documentation.SBV.Examples.Crypto.SHA & 108 & 44 & 5830 & 49 & 18 & 3.69\\\hline
Documentation.SBV.Examples.Existentials.CRCPolynomial & 10 & 44 & 360 & 44 & 12 & 3.07\\\hline
Documentation.SBV.Examples.Existentials.Diophantine & 26 & 44 & 682 & 44 & 8 & 5.39\\\hline
Documentation.SBV.Examples.Lists.BoundedMutex & 19 & 44 & 1375 & 44 & 22 & 3.21\\\hline
Documentation.SBV.Examples.Lists.Fibonacci & 4 & 44 & 225 & 44 & 5 & 3.36\\\hline
Documentation.SBV.Examples.Lists.Nested & 1 & 44 & 444 & 44 & 4 & 3.63\\\hline
Documentation.SBV.Examples.Misc.Auxiliary & 4 & 44 & 155 & 44 & 5 & 3.38\\\hline
Documentation.SBV.Examples.Misc.Enumerate & 7 & 44 & 205 & 44 & 11 & 3.26\\\hline
Documentation.SBV.Examples.Misc.Floating & 13 & 44 & 587 & 44 & 9 & 3.04\\\hline
\end{tabular}

\begin{tabular}{|l|l|l|l|l|l|l|}
    \hline
Documentation.SBV.Examples.Misc.ModelExtract & 3 & 44 & 141 & 44 & 4 & 3.77\\\hline
Documentation.SBV.Examples.Misc.Newtypes & 3 & 44 & 95 & 44 & 7 & 3.07\\\hline
Documentation.SBV.Examples.Misc.NoDiv0 & 3 & 44 & 113 & 43 & 5 & 3.70\\\hline
Documentation.SBV.Examples.Misc.Polynomials & 9 & 44 & 209 & 43 & 5 & 3.26\\\hline
Documentation.SBV.Examples.Misc.SetAlgebra & 0 & 0 & 0 & 0 & 0 & 3.53\\\hline
Documentation.SBV.Examples.Misc.SoftConstrain & 1 & 44 & 178 & 44 & 4 & 3.26\\\hline
Documentation.SBV.Examples.Misc.Tuple & 16 & 44 & 337 & 44 & 5 & 4.12\\\hline
Documentation.SBV.Examples.Optimization.Enumerate & 12 & 44 & 201 & 44 & 5 & 3.20\\\hline
Documentation.SBV.Examples.Optimization.ExtField & 1 & 44 & 130 & 44 & 8 & 2.77\\\hline
Documentation.SBV.Examples.Optimization.LinearOpt & 2 & 44 & 224 & 44 & 5 & 3.60\\\hline
Documentation.SBV.Examples.Optimization.Production & 7 & 44 & 205 & 44 & 4 & 3.65\\\hline
Documentation.SBV.Examples.Optimization.VM & 8 & 44 & 568 & 44 & 8 & 3.24\\\hline
Documentation.SBV.Examples.ProofTools.BMC & 5 & 44 & 140 & 44 & 9 & 3.89\\\hline
Documentation.SBV.Examples.ProofTools.Fibonacci & 11 & 44 & 297 & 44 & 10 & 3.12\\\hline
Documentation.SBV.Examples.ProofTools.Strengthen & 11 & 44 & 456 & 44 & 17 & 2.97\\\hline
Documentation.SBV.Examples.ProofTools.Sum & 9 & 44 & 207 & 44 & 10 & 3.17\\\hline
Documentation.SBV.Examples.Puzzles.Birthday & 21 & 44 & 807 & 44 & 8 & 6.62\\\hline
Documentation.SBV.Examples.Puzzles.Coins & 14 & 44 & 513 & 44 & 10 & 3.01\\\hline
Documentation.SBV.Examples.Puzzles.Counts & 12 & 44 & 481 & 44 & 5 & 2.75\\\hline
Documentation.SBV.Examples.Puzzles.DogCatMouse & 6 & 44 & 207 & 44 & 6 & 3.31\\\hline
Documentation.SBV.Examples.Puzzles.Euler185 & 11 & 44 & 609 & 44 & 4 & 3.53\\\hline
Documentation.SBV.Examples.Puzzles.Fish & 48 & 44 & 1029 & 44 & 7 & 3.43\\\hline
Documentation.SBV.Examples.Puzzles.Garden & 12 & 44 & 494 & 44 & 10 & 3.64\\\hline
Documentation.SBV.Examples.Puzzles.HexPuzzle & 21 & 44 & 1187 & 45 & 15 & 3.97\\\hline
Documentation.SBV.Examples.Puzzles.LadyAndTigers & 3 & 44 & 277 & 44 & 4 & 3.48\\\hline
Documentation.SBV.Examples.Puzzles.MagicSquare & 19 & 44 & 571 & 44 & 6 & 3.44\\\hline
Documentation.SBV.Examples.Puzzles.NQueens & 10 & 44 & 411 & 44 & 4 & 4.09\\\hline
Documentation.SBV.Examples.Puzzles.SendMoreMoney & 4 & 44 & 361 & 44 & 4 & 2.87\\\hline
Documentation.SBV.Examples.Puzzles.Sudoku & 42 & 44 & 3980 & 44 & 5 & 3.69\\\hline
Documentation.SBV.Examples.Puzzles.U2Bridge & 56 & 44 & 1611 & 44 & 14 & 3.60\\\hline
Documentation.SBV.Examples.Queries.AllSat & 7 & 44 & 455 & 44 & 6 & 3.27\\\hline
Documentation.SBV.Examples.Queries.CaseSplit & 2 & 44 & 420 & 48 & 4 & 3.54\\\hline
Documentation.SBV.Examples.Queries.Concurrency & 8 & 44 & 1485 & 44 & 9 & 3.36\\\hline
Documentation.SBV.Examples.Queries.Enums & 8 & 44 & 234 & 44 & 12 & 3.55\\\hline
Documentation.SBV.Examples.Queries.FourFours & 29 & 44 & 1156 & 46 & 21 & 3.48\\\hline
Documentation.SBV.Examples.Queries.GuessNumber & 10 & 44 & 386 & 44 & 5 & 2.87\\\hline
Documentation.SBV.Examples.Queries.Interpolants & 2 & 44 & 431 & 44 & 5 & 3.71\\\hline
Documentation.SBV.Examples.Queries.UnsatCore & 2 & 44 & 248 & 44 & 4 & 3.25\\\hline
Documentation.SBV.Examples.Strings.RegexCrossword & 16 & 44 & 784 & 44 & 10 & 3.56\\\hline
Documentation.SBV.Examples.Strings.SQLInjection & 9 & 44 & 727 & 45 & 13 & 4.18\\\hline
Documentation.SBV.Examples.Transformers.SymbolicEval & 29 & 44 & 852 & 46 & 8 & 6.51\\\hline
Documentation.SBV.Examples.Uninterpreted.AUF & 6 & 44 & 219 & 45 & 7 & 2.80\\\hline
Documentation.SBV.Examples.Uninterpreted.Deduce & 7 & 44 & 244 & 44 & 6 & 3.46\\\hline
\end{tabular}

\begin{tabular}{|l|l|l|l|l|l|l|}
    \hline
Documentation.SBV.Examples.Uninterpreted.Function & 2 & 44 & 52 & 43 & 5 & 3.36\\\hline
Documentation.SBV.Examples.Uninterpreted.Multiply & 5 & 44 & 215 & 44 & 9 & 3.24\\\hline
Documentation.SBV.Examples.Uninterpreted.Shannon & 18 & 44 & 489 & 43 & 9 & 3.17\\\hline
Documentation.SBV.Examples.Uninterpreted.Sort & 3 & 44 & 129 & 44 & 4 & 3.20\\\hline
Documentation.SBV.Examples.Uninterpreted.UISortAllSat & 3 & 44 & 220 & 44 & 4 & 3.18\\\hline
Documentation.SBV.Examples.WeakestPreconditions.Append & 10 & 44 & 538 & 47 & 6 & 3.19\\\hline
Documentation.SBV.Examples.WeakestPreconditions.Basics & 8 & 44 & 204 & 47 & 13 & 3.73\\\hline
Documentation.SBV.Examples.WeakestPreconditions.Fib & 13 & 44 & 573 & 47 & 7 & 3.22\\\hline
Documentation.SBV.Examples.WeakestPreconditions.GCD & 17 & 44 & 604 & 47 & 7 & 2.71\\\hline
Documentation.SBV.Examples.WeakestPreconditions.IntDiv & 12 & 44 & 447 & 47 & 13 & 3.65\\\hline
Documentation.SBV.Examples.WeakestPreconditions.IntSqrt & 14 & 44 & 509 & 47 & 13 & 2.95\\\hline
Documentation.SBV.Examples.WeakestPreconditions.Length & 11 & 44 & 305 & 47 & 8 & 2.75\\\hline
Documentation.SBV.Examples.WeakestPreconditions.Sum & 9 & 44 & 369 & 47 & 13 & 3.59\\\hline
\end{tabular}

  \subsection{Package \texttt{time-1.10}}

  \begin{tabular}{|l|l|l|l|l|l|l|}
\hline
Name & N & K & V & D & I & Time (ms)\\\hline
Data.Format & 33 & 3 & 711 & 1 & 2 & 3.85\\\hline
Data.Time & 0 & 0 & 0 & 0 & 0 & 2.48\\\hline
Data.Time.Calendar & 0 & 0 & 0 & 0 & 0 & 2.92\\\hline
Data.Time.Calendar.CalendarDiffDays & 7 & 1 & 25 & 1 & 2 & 3.74\\\hline
Data.Time.Calendar.Days & 3 & 0 & 12 & 0 & 2 & 3.85\\\hline
Data.Time.Calendar.Easter & 7 & 1 & 102 & 1 & 3 & 2.62\\\hline
Data.Time.Calendar.Gregorian & 35 & 1 & 283 & 1 & 4 & 3.55\\\hline
Data.Time.Calendar.Julian & 35 & 1 & 283 & 1 & 4 & 2.94\\\hline
Data.Time.Calendar.JulianYearDay & 11 & 1 & 109 & 1 & 1 & 3.24\\\hline
Data.Time.Calendar.MonthDay & 9 & 0 & 220 & 0 & 0 & 2.76\\\hline
Data.Time.Calendar.OrdinalDate & 26 & 1 & 371 & 1 & 2 & 2.62\\\hline
Data.Time.Calendar.Private & 19 & 2 & 118 & 1 & 1 & 3.70\\\hline
Data.Time.Calendar.Week & 1 & 7 & 14 & 1 & 2 & 2.58\\\hline
Data.Time.Calendar.WeekDate & 16 & 1 & 233 & 1 & 2 & 2.85\\\hline
Data.Time.Clock & 0 & 0 & 0 & 0 & 0 & 3.72\\\hline
Data.Time.Clock.Internal.AbsoluteTime & 4 & 1 & 27 & 2 & 3 & 2.93\\\hline
Data.Time.Clock.Internal.CTimespec & 8 & 1 & 283 & 1 & 4 & 4.02\\\hline
Data.Time.Clock.Internal.CTimeval & 2 & 1 & 90 & 1 & 2 & 3.46\\\hline
Data.Time.Clock.Internal.DiffTime & 3 & 0 & 10 & 0 & 1 & 3.74\\\hline
Data.Time.Clock.Internal.NominalDiffTime & 3 & 0 & 7 & 0 & 1 & 2.84\\\hline
Data.Time.Clock.Internal.POSIXTime & 1 & 0 & 1 & 0 & 1 & 2.29\\\hline
Data.Time.Clock.Internal.SystemTime & 8 & 1 & 74 & 1 & 3 & 2.63\\\hline
Data.Time.Clock.Internal.UTCDiff & 2 & 1 & 14 & 3 & 3 & 2.41\\\hline
Data.Time.Clock.Internal.UTCTime & 2 & 1 & 6 & 1 & 2 & 3.83\\\hline
Data.Time.Clock.Internal.UniversalTime & 1 & 0 & 3 & 0 & 1 & 3.74\\\hline
Data.Time.Clock.POSIX & 6 & 1 & 62 & 3 & 3 & 2.97\\\hline
Data.Time.Clock.System & 9 & 1 & 143 & 3 & 3 & 2.62\\\hline
Data.Time.Clock.TAI & 8 & 1 & 167 & 3 & 9 & 2.73\\\hline
Data.Time.Format & 0 & 0 & 0 & 0 & 0 & 2.85\\\hline
Data.Time.Format.Format.Class & 24 & 2 & 614 & 3 & 7 & 4.04\\\hline
Data.Time.Format.Format.Instances & 3 & 1 & 29 & 2 & 2 & 2.85\\\hline
Data.Time.Format.ISO8601 & 66 & 3 & 2137 & 6 & 10 & 2.81\\\hline
Data.Time.Format.Internal & 0 & 0 & 0 & 0 & 0 & 2.90\\\hline
Data.Time.Format.Locale & 11 & 1 & 605 & 1 & 2 & 2.63\\\hline
Data.Time.Format.Parse & 14 & 1 & 285 & 2 & 2 & 2.92\\\hline
Data.Time.Format.Parse.Class & 25 & 3 & 1375 & 2 & 6 & 2.99\\\hline
Data.Time.Format.Parse.Instances & 17 & 1 & 506 & 2 & 2 & 2.81\\\hline
Data.Time.LocalTime & 0 & 0 & 0 & 0 & 0 & 3.44\\\hline
Data.Time.LocalTime.Internal.CalendarDiffTime & 5 & 1 & 29 & 2 & 3 & 2.97\\\hline
Data.Time.LocalTime.Internal.LocalTime & 14 & 1 & 122 & 3 & 4 & 4.71\\\hline
Data.Time.LocalTime.Internal.TimeOfDay & 24 & 1 & 294 & 1 & 3 & 2.83\\\hline
Data.Time.LocalTime.Internal.TimeZone & 16 & 2 & 343 & 3 & 3 & 2.79\\\hline
Data.Time.LocalTime.Internal.ZonedTime & 6 & 1 & 46 & 5 & 4 & 2.96\\\hline
\end{tabular}

  \subsection{Package \texttt{unordered-containers-0.2.11.0}}

  \begin{tabular}{|l|l|l|l|l|l|l|}
\hline
Name & N & K & V & D & I & Time (ms)\\\hline
Data.HashMap.Array & 78 & 1 & 1468 & 1 & 6 & 2.99\\\hline
Data.HashMap.Base & 269 & 5 & 4509 & 3 & 24 & 2.82\\\hline
Data.HashMap.Lazy & 0 & 0 & 0 & 0 & 0 & 2.68\\\hline
Data.HashMap.List & 12 & 0 & 158 & 0 & 0 & 3.30\\\hline
Data.HashMap.Strict & 0 & 0 & 0 & 0 & 0 & 2.98\\\hline
Data.HashMap.Strict.Base & 89 & 5 & 1359 & 3 & 14 & 2.75\\\hline
Data.HashMap.Unsafe & 0 & 0 & 0 & 0 & 0 & 2.85\\\hline
Data.HashMap.UnsafeShift & 2 & 0 & 20 & 0 & 0 & 3.53\\\hline
Data.HashSet & 0 & 0 & 0 & 0 & 0 & 2.84\\\hline
Data.HashSet.Base & 24 & 5 & 247 & 3 & 5 & 3.82\\\hline
\end{tabular}

\end{toappendix}

\section{Related work}\label{sec:related}
The goal of our system is to automatically, statically verify that a given program is free of pattern match exceptions, and we have phrased it as a type inference procedure for a certain refinement type system with recursive datatype constraints.
We have shown that it works well in practice, although a more extensive investigation is needed.
Our primary motivation has been to ensure predictability by giving concrete guarantees on its expressive power and algorithmic complexity.

\paragraph{Recursive types, subtyping and set constraints}
Our work sits within a large body of literature on recursive types and subtyping.
As a type system, ours is not directly comparable to others in the literature: on the one hand, the intensional refinement restriction is quite severe, but on the other we allow for path sensitivity.
One of the first works to consider subtyping in the setting of recursive types was that of \citet*{RN48}.
They proposed an exponential time procedure for subtype checking, but this was later improved to quadratic by \citet*{RN40}.
Neither of these works gave a treatment of the combination with polymorphism, which is the subject of e.g. \citet{RN38,RN51,RN3,RN47}.
However, to the best of our knowledge, all the associated type inference algorithms are exponential time in the size of the program.
In particular, \citet{RN51} shows that a general formulation of typing with recursive subtyping constraints has a PSPACE-hard typability problem.
However, we mention as a counterpoint that when constraints are restricted to simple variable-variable inequalities, \citet*{RN52} show that there is a cubic-time algorithm.
Being based on unification, inference for polymorphic variants is efficient \cite{RN45}, but \citet*{RN1} point out instances where programmers find the results to be unpredictable.
None of the above allow for path-sensitive treatment of matching.

Our main inspiration has been the seminal body of literature of work on set constraints in program analysis, see particularly \citet*{RN41}, \citet*{RN15} and \citet*{RN14}, and in particular, the line of work on making the cubic-time fragments scale in practice \cite{RN66,RN49,RN63,RN64}.
Through an impressive array of sophisticated optimisations, the fragment can be made to run efficiently on many programs.  However, the fundamental worst-case complexity is not changed and implementing and tuning heuristics requires a large engineering effort.  Moreover, this fragment does not accommodate path sensitivity.

An interesting new approach to full set constraints language is that of \citet*{eremondi-tfp19}, who attempts to use SMT to circumvent the extremely high worst-case complexity in some practical cases.  However, experiments are limited to programs less than a few hundred lines.

Many of the analyses and or type inference procedures discussed so far are compositional, i.e. parts of the program are analysed independently to yield summaries of their behaviour and then the summaries are later combined.   However, it has been frequently observed that compositionality does not lead to scalability if the summaries are themselves large and complicated.
In particular, it is not uncommon for ``summaries'' that grow with the square of the size of the program in the worst case.
This has led to many works that attempt to simplify summaries, typically according to ingenious heuristics \cite{RN54,RN69,RN70,RN50,RN68,RN38,RN49,pottier-njc2000}.
Since our primary motivation was predictability, we have designed our system so that heuristics are avoided\footnote{Heuristic-based optimisations can be the enemy of predictability since small changes in the program can lead to great changes in performance if the change causes the program to fall outside of the domain on which the heuristic is tuned.}: in particular the size of summaries (i.e. constrained type schemes) only depends only on the size of the underlying types and not the size of the program.
It is plausible that many of these heuristic optimisations are  nevertheless applicable in order to help improve the overall efficiency.
Note also that, if we are not concerned with a compositional analysis, then our class of constraints can be checked for solvability using a linear time algorithm due to \citet*{rehof-mogensen-scp1999}.  However, as explained in the introduction, compositionality is essential to obtaining  \emph{overall} linear time complexity.

\paragraph{Refinement types}
Refinement types originate with the works of \citet*{RN37} and \citet*{RN55}.
Their distinguishing feature is that they attempt to assign types to program expressions for which an \emph{underlying type} is already available.
Typically, as here, the refinement type is also required to respect the shape of the underlying type.
One can use this restriction, as in \emph{loc cit} to ensure some independence of the the size of the type from the size of the program.
However, as remarked in the final section, the constant factors are enormous since there is unrestricted intersection and union of refinements of the same underlying type which is represented explicitly.

The work of \citet{RN37} requires that the programmer declare the universe of refinement types up-front (where our universe is determined automatically as a completion of the underlying datatype environment).
A disadvantage of this requirement is that it burdens the programmer with a kind of annotation that they would rather not have to clutter their program with, in many simple cases.
A great advantage is that, by defining a refinement datatype explicitly, the programmer can indicate formally in the code her intention that a certain invariant is (somehow) important within a certain part of the program.
It seems like a very fruitful idea to allow the programmer this freedom also in our system and we are actively working on an extension to allow for this as part of our future work.
In particular, we would like to take advantage of several new advances in this line that relieve a lot of programmer burden, such as those of \citet*{RN34,RN35}.

An incredibly fruitful recent evolution of refinement types are the Liquid Types of \citet*{RN30} (see especially \citet*{RN32} for a version with constrained type schemes) and similar systems (e.g. those of \citet*{RN58,RN56}).
Such technology is already accessible to the benefit of the average programmer through the Liquid Haskell system of \citet*{RN28}.
Due to the rich expressive power of these systems, which typically include dependent product, efficient and fully-automatic type inference is not typically a primary concern and predictability can be ensured by liberal use of annotations.

\paragraph{Pattern match safety and model checking}
The pattern match safety problem was also addressed by \citet*{RN26}, which was used to verify a number of small Haskell programs and libraries.
The expressive power and algorithmic complexity are, however, unclear.

Safety problems are within the scope of higher-order model checking (\citet{RN21,RN20,RN22}) and a system for verifying pattern match safety, built on higher-order model checking was presented in \cite{RN23}.
Higher-order model checking approaches reduce verification problems to model checking problems on a certain infinite tree generated by a higher-order grammar.
Although the higher-order model checking problem is linear-time in the size of the grammar, the constant factors are enormous because, formally, it is $n$-EXPTIME complete (tower of exponentials of height $n$) with $n$ the type-theoretic order of the functions in the grammar.
Moreover, many of the transformations from program to grammar incur a large blow-up in size.
Two promising evolutions of higher-order model checking are the approach of \citet*{RN61} based on the higher-order fixpoint logic of \citet{RN60} and that of \citet*{RN62} via higher-order constrained Horn clauses.

\paragraph{Contract checking}
Like pattern-match safety, static contract checking problems such as those considered by \citet*{xu-et-al-popl09}, \citet*{vytiniotis-et-al-popl13} and \citet*{nguyen-etal-icfp2014} typically also reduce to reachability.  However, giving guarantees on scalability via worst-case complexity does not seem to be a priority for this area and experiments are correspondingly limited to programs of only a few hundred lines.

\paragraph{Pattern match coverage checking}
A related problem is the pattern match \emph{coverage checking} problem, which asks, with respect to the type of the function: if a given set of patterns is exhaustive, non-overlapping and irredundant  (a classic paper on this subject is that of \citet*{maranget-jfp07}, but see \citet*{graf-et-al-icfp20} for more recent developments).
To illustrate the difference between the two problems: a program containing the following definition is always a no-instance of the coverage checking problem, since $f :: \mathsf{Bool} \to \mathsf{Bool}$ has non-exhaustive patterns:
\[
\begin{array}{l}
    f\ \mathsf{True} = \mathsf{False}
\end{array}
\]
However, such a program may or may not be a no-instance of the pattern-match safety problem, since it depends on how $f$ is actually used in the rest of the program.
Like all safety problems, the latter can be reduced to reachability: is there an execution of the program that reaches a call of the form $f\ \mathsf{False}$.
If no execution of the program containing this definition ever calls $f$ with $\mathsf{False}$, then this program will be a yes-instance of the safety problem, i.e. pattern-match safe.

On the one-hand, if every pattern-matching expression covers all cases, then the program is already safe, since no execution can trigger a pattern-match violation.
On the other, a program may be safe and yet not cover every case in its patterns --- indeed these are really the focus from a program verification perspective.
The proliferation of exotic kinds of pattern allowed in a complex language such as Haskell (e.g. pattern synonyms \cite{pickering-et-al-Haskell16}), mean that coverage checking may sometimes benefit from reasoning about program executions in a localised way.  However, as the authors of \cite{graf-et-al-icfp20} point out, it is ``unreasonable to expect a coverage checking algorithm to prove [a property of arbitrary program executions]''.
\section{Conclusion}\label{sec:conclusion}

We have presented a new extension of ML-style typing with  intensional refinements of algebraic datatypes. 
Since type inference is fully automatic, the system can be used as a program analysis for verifying the pattern-match safety problem.
Viewed this way, it incorporates polyvariance and path-sensitivity and yet we have shown that, under reasonable assumptions, the worst-case time complexity is linear in the size of the program.
To achieve this we have shifted exponential complexity associated with HM(X)-style inference from the size of the program to the \emph{size of the types of the program}.
Moreover, we have shown that our assumptions on the size of types are reasonable in practice (equivalently: that the constant factors are not prohibitive) by demonstrating excellent performance of a prototype.


This was only possible because we have made a compromise on the space of invariants that we can synthesize: typings are built from \emph{intensional} refinements.
Although our analysis is polyvariant and path-sensitive, these features are ultimately limited by the shape of these refinements.
We have given examples, (e.g. in Figure \ref{fig:fm-refns} and Example \ref{ex:lam-refinements}) of datatypes for which there are many useful intensional refinements that are expressible in our system.
However, there are also common datatypes for which there are no useful intensional refinements.
For example, the four intensional refinements of the datatype of lists correspond to: the empty type, the type of infinite lists, the type containing only the empty list and the original list datatype.
However, none of these is especially useful in practice, and one would much rather have a refinement like the type of non-empty lists.

In any fully automatic program analysis, there will always be some compromise on expressivity.
We believe it is important that one can understand, before using the tool on a program, whether the compromise will be a real limitation.
The power of our analysis is characterised as a type system that can be understood by programmers familiar with usual ML-style typing.
If the (user believes that their) program would be typable in this system (i.e. there exist intensional refinements of the datatypes under which a typing can be assigned), then the analysis will be able to verify it.
Note that the user does not need to know anything about constraints, which occur only in the inference algorithm, in order to determine this.
For example, if the user believes that the safety of their program relies on a invariant to do with the non-emptiness of lists then, since non-emptiness is not an expressible refinement of lists, they should not expect the program to be verifiable.

Our future work concerns such cases.  
We would like to enable the user to specify their own non-intensional refinements of datatypes to extend the space of expressible program invariants.  
For example, using some syntax, the programmer could indicate that the non-emptiness of lists is an important refinement.
Under the hood, the system can extend the original definition of lists in such a way that (a) the new definition is extensionally equivalent to the original, but (b) non-empty lists is now an intensional refinement.
The following is an example of such a redefinition:
\begin{lstlisting}
    List  a = Nil | Cons a (List' a) 
    List' a = Nil | Cons a (List' a)
    £
\end{lstlisting}
From which the non-empty list refinement arises by erasing \textsf{Nil} from the \textsf{List} datatype.
This would put the trade-off between expressive power and efficiency in the hands of the user, and since the type system is familiar, they are well equipped to reason about when it makes sense.

\begin{acks}                            
We gratefully acknowledge the support of the  
\grantsponsor{EPSRC}{Engineering and Physical Sciences Research Council}{http://https://epsrc.ukri.org} (\grantnum{EPSRC}{EP/T006579/1}) and the National Centre for Cyber Security via the UK Research Institute in Verified Trustworthy Software Systems.
We thank our colleague Matthew Pickering for a lot of good Haskell advice and for helping us safely navigate the interior of the Glasgow Haskell Compiler.
\end{acks}

\clearpage

\bibliography{references,endnote}

\ifsupp
  \appendix
\else
\fi

\end{document}